\documentclass[journal]{IEEEtran}
\usepackage{graphicx,amsthm,bm,amsmath}
\usepackage{siunitx}
\usepackage{graphicx,amsthm,bm,algorithm,algorithmicx,algpseudocode,multirow,tabularx,array}
\graphicspath{{./figures/}}
\usepackage{color,soul,textcomp}
\usepackage[dvipsnames]{xcolor}

\newtheorem{Lem}{Lemma}
\newtheorem{Theo}{Theorem}
\newcolumntype{C}[1]{>{\centering\let\newline\\\arraybackslash}m{#1}}

\begin{document}
\title{From Stochastic to Bit Stream Computing: Accurate Implementation of Arithmetic Circuits and Applications in Neural Networks}	
\author{Ensar~Vahapoglu and~Mustafa~Altun
\IEEEcompsocitemizethanks{
--------------------------------------------------------------------------------------------

\textcopyright 2019 IEEE. Personal use of this material is permitted. Permission from IEEE must be obtained for all other uses, in any current or future media, including reprinting/republishing this material for advertising or promotional purposes, creating new collective works, for resale or redistribution to servers or lists, or reuse of any copyrighted component of this work in other works. \protect\\

\textbf{*} This work is supported by the TUBITAK-1001 project \#116E250 and Istanbul Technical University BAP (ITU-BAP) project \#40781. \protect\\

\vspace{-1,5mm}
\textbf{*} E. Vahapoglu and M. Altun are with the Department of Electronics and Communication Eng., Istanbul Technical University, Istanbul,
Turkey, 34469.\protect\\

\vspace{-1,5mm}
\textbf{*} E-mails: \{vahapoglu, altunmus\}@itu.edu.tr \protect\\

\vspace{-1,5mm}
\textbf{*} A preliminary version of this paper appeared in \cite{vahapoglu2016accurate} 
}}

\IEEEcompsoctitleabstractindextext{
\begin{abstract}
In this study, we propose a novel computing paradigm ``Bit Stream Computing" that is constructed on the logic used in stochastic computing, but does not necessarily employ randomly or Binomially distributed bit streams as stochastic computing does. Any type of streams can be used either stochastic or deterministic. The proposed paradigm benefits from the area advantage of stochastic logic and the accuracy advantage of conventional binary logic. 
We implement accurate arithmetic multiplier and adder circuits, classified as asynchronous or synchronous; we also consider their suitability of processing successive streams.  
The proposed circuits are simulated using the Cadence Genus tool with TSMC \SI{}{\textbf{0.18\micro m}} CMOS technology. 
We thoroughly compare the proposed adders and multipliers with their predecessors in the literature, individually and in a neural network application. Comparisons are made in terms of area, speed, power, and accuracy.
We believe that this study opens up new horizons for computing that enables us to implement much smaller yet accurate arithmetic circuits compared to the conventional binary and stochastic ones.

\end{abstract}

\begin{IEEEkeywords}
	Stochastic computing, bit stream computing, arithmetic circuits, neural network.
\end{IEEEkeywords}}

\maketitle
\IEEEdisplaynotcompsoctitleabstractindextext

\section{Introduction} \label{Intro}
\IEEEPARstart{s}{tochastic} computing (SC), first brought forward in 1960s  \cite{von1956probabilistic, gaines1967stochastic}, performs serial data processing with Binomially distributed bit streams. 
Each stream represents a probability value, obtained as the number of 1 valued bits over the total number of bits. Thus, it is possible to use $n+1$ different states with a single input/output, corresponding to $n+1$ different values ranging from $0/n$ to $n/n$  where $n$ is the total number of bits in a stream. On the other hand, conventionally a binary input/output has two states that are logic 0 or logic 1. This  feature offers an important area advantage for SC, especially for arithmetic operations. For example, a single AND gate is used for stochastic multiplication. This is illustrated in Fig. \ref{fig_ClassicalStochastic}. Here, input streams have values of 1/2, so an output value of 1/4 is expected. Although the correct result can be achieved as in Fig. \ref{fig_ClassicalStochastic} a), it is not guaranteed for different cases since 1's and 0's in streams are randomly positioned in SC. Fig. \ref{fig_ClassicalStochastic} b) shows an erroneous result with an output value of 0/4. Here, the relative standard error is 100\%. Note that we represent streams such that the bit on the leftmost is the first to be processed.  

\begin{figure}[!t]
	\centering
	\includegraphics{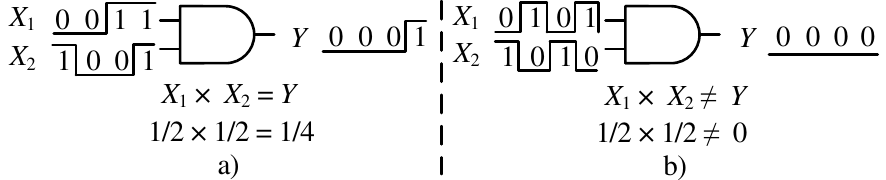}
	\caption{Stochastic multiplication with an AND gate having: a) accurate results, and b) inaccurate results.}
	\label{fig_ClassicalStochastic}
\end{figure}

\begin{figure}[!t]
	\centering
	\includegraphics[scale=0.85]{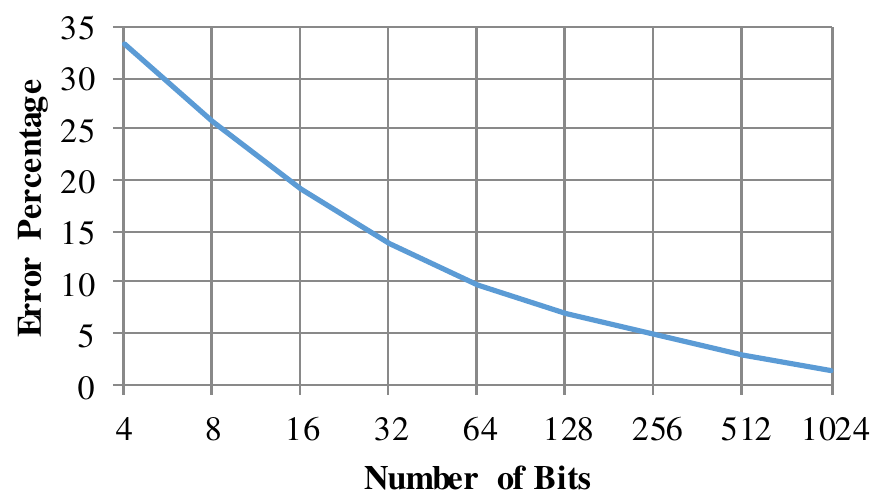}
	\caption{Average error percentage for an AND gate with respect to the number of bits ($n$) in a stream; inputs probability values are both $ 1/2 $.}
	\label{fig_AvErrorAnd}
\end{figure}

Indeed, SC cannot guarantee error-free computation due to its random feature. Since streams are Binomially distributed, to always achieve zero error, infinite number of bits are needed. Fig. \ref{fig_AvErrorAnd} shows how the average error changes with the number of bits for an AND gate having input values of 1/2. Here, to achieve 10\% and 1\% errors, streams having more than 100 and 1000 bits are needed that is not practical in terms of the computing time. This explains why SC could not become a real competitor to conventional computing although it offers significant area advantage \cite{alaghi2013survey}. Low accuracy or long computing times is the main obstacle in front of SC and the main motivation of this study. 

In this paper, we propose a novel computing paradigm ``Bit Stream Computing (BSC)". Similar to SC, BSC uses unary bit streams having time series of 0's and 1's, and the value of a bit stream is calculated by the total number of 1 valued bits over the total number of bits in the stream. Different from SC, BSC does not necessarily employ randomly or Binomially distributed input/output bit streams and accordingly the bit stream values do not necessarily represent probability values. For example, an AND gate is used for multiplication in SC because applying the logic AND operation to two independent probability values $p_1$ and $p_2$ results in a probability value $p_1\times p_2$. However, an AND gate is not proper for multiplication in BSC since it should be guaranteed in BSC that same input values always result in same output values, and it is not satisfied with an AND gate for multiplication. BSC does not require any type of distribution for streams. In other words, changing the orderings of 1's and 0's in input bit streams should not alter output values. 
The proposed paradigm benefits from the area advantage of stochastic logic and the accuracy advantage of conventional binary logic. With BSC, we successfully implement accurate arithmetic multiplier and adder circuits.

Along with the accuracy issue  that is kept on the agenda since the birth of SC, there is another difficulty for SC as well as for BSC: timing problems. They happen mainly due to undesirable changes in the duration of 1's and 0's in a bit stream. Three examples are given in Fig. \ref{fig_RealizationErrors} showing the expected error-free and obtained erroneous signal forms. In the first example in Fig. \ref{fig_RealizationErrors} a), the signal is not processed fast enough; in the second example in Fig. \ref{fig_RealizationErrors} b), there is a glitch; and in the third example in Fig. \ref{fig_RealizationErrors} c), the difference between rise and fall times of the signal causes an error. Considering the severity of these and similar types of timing problems,  any circuit design technique developed for SC or BSC should be justified with timing considerations. For this purpose, we test the proposed circuits with transistor level simulations by considering the correctness of output signals with their values as well as the integrity of the signals showing how close the signals are to ideal forms. 

\begin{figure}[!t]
	\centering
	\includegraphics[scale=1]{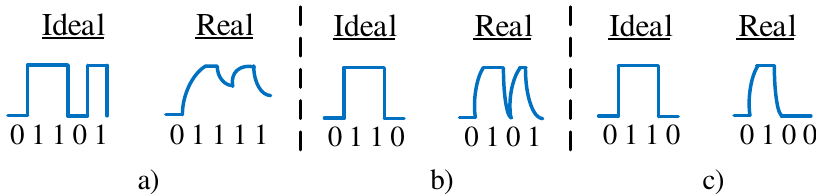}	
	\caption{Examples of timing problems caused by: a) slow processing, b) glitch, and c) difference in rise and fall times.}
	\vspace{-10pt}
	\label{fig_RealizationErrors}
\end{figure}

\subsection{Previous Works and Contributions }
The mainstream solution to improve accuracy in SC is manipulating input bit streams by either decreasing their randomness or making them dependent/correlated. 
For this purpose, pseudo-random and quasi-random number generators are proposed \cite{ichihara2014compact,gupta1988binary}. Pseudo-random generators generally use LFSR's (Linear Feedback Shift Register) that even allow to produce desired orders of 0's and 1's resulting in perfect accuracy \cite{gupta1988binary}. Additionally, quasi-random generators producing low-discrepancy bit streams can decrease error rates \cite{alaghi2014fast, alaghi2015logic}. There are also recent studies exploiting correlation for accuracy \cite{alaghi2013exploiting}, as well as using fully deterministic generators \cite{jenson2016deterministic, najafi2017time}. In \cite{jenson2016deterministic} and \cite{najafi2017time}, accurate arithmetic operations are achieved. 
Nevertheless, for all of these studies uncorrelated or independent bit streams are needed for all inputs. This is achieved either using a separate stream generator for each input or sharing a generator  for multiple inputs with extra circuitries \cite{ichihara2014compact,najafi2017time}. Although the sharing allows some area saving, still the total area needed to generate input streams is linearly dependent with the number of inputs, and this area consumes a majority of the circuit area. 

Another important drawback of the mentioned studies is that they are not suitable for multi-level designs. Outputs of one level can not be directly used as inputs of another level. For example, outputs of two AND gates can not be directly used as inputs of another AND gate. Outputs should be recreated to fit the desired format, and this is quite costly. In \cite{najafi2017time} the authors discuss this problem. The generated input signals, called as PWM signals in the paper, loose their formation at the output. They propose a solution for this, but it requires extra control inputs, so new streams are needed to be generated. This does not just worsen the design complexity, but it also decreases the speed dramatically.  

As opposed to the studies focusing on the generation of bit streams in desired formats, our treatment BSC does only care about the values carried by the streams, so any type of input bit streams can be directly used. This eliminates the need of specific stream generators. Furthermore, there is no extra cost for multi-level designs; output streams of one level can be directly used as inputs. In the literature,  using the same logic, an accurate adder is proposed in \cite{lee2017energy}, called as Alaghi adder in the paper. One of our two proposed synchronous adders is quite similar to the Alaghi adder with additional considerations for timing. Also we show that our adder can be generalized for any number of inputs. Furthermore, along with the adders, we propose two synchronous multipliers.  

All of the above mentioned designs with an aim of improving accuracy use clock signals, so they are synchronous. To eliminate the cost of synchronization, we also propose asynchronous adders and multipliers mainly constructed on delay elements. 
Another shortcoming of these designs is their inability to process successive input bit streams; they are assumed to perform one-time operations. To overcome this shortcoming, we propose an adder and a multiplier that can successively process input bit streams. 

Apart from the mentioned shortcomings, underestimating
the timing problems is a general tendency in
the literature. These problems are indigenous to bit stream operations in SC and BSC, and without solving them it is hard to claim the feasibility of the proposed study. For example, in \cite{najafi2017time}, the authors claim to work with 1 GHz clock signals to generate pulses as input streams. Suppose that 8 bit binary equivalent operations are performed, so there should be at least 256 different values for input streams. Therefore, for the worst case scenario to represent the value of 1/256, a bit stream or a pulse has a 1 valued bit with a duration of 1/($256\times 10^9$) seconds. It means that the proposed circuits should safely process 0.256 THz signals that does not seem to be possible (recall the case in Fig. \ref{fig_RealizationErrors} a)).
Therefore, much slower operations should be used that also causes dramatic area increase for this study (justified in the experimental results section). 

Indeed, timing problems have high significance for any computing paradigm using time series of bits including SC, BSC, and bit serial computing. 
A general solution is using a latch for each output of a circuit block, so time durations of 1 and 0 valued signals are kept close to expected ideal values  \cite{hartley2012digit,parhami2010computer}.  Although this clock based timing is precise, it might cause a considerably large area overhead and it is not suitable for asynchronous designs. Another less precise solution is using buffers \cite{alaghi2017trading}. Of course, increasing time durations of bits also helps solving the timing problems at the cost of decreased speed. Although, solutions to timing problems are not in the scope of this paper, we at least aim to show the timing performance of the proposed designs. For this reason, we have designed all of the proposed circuits in transistor level with timing simulations. 

\subsection{Overview}

Three adders and three multipliers performing BSC are proposed. Among the proposed six circuits, one adder and one multiplier  are asynchronous, and the rest four are synchronous. Among the four synchronous circuits, one adder and one multiplier are able to process successive input bit streams.  We evaluate all of the proposed designs with their predecessors by performing simulations with TSMC 0.18\SI{}{\micro m} CMOS technology.  The proposed circuits are also tested in a neural network application.


The rest of paper is lined up as follows. Section \ref{limits} is composed of definitions, explanations, and limitations for BSC. In Section \ref{Async} and Section \ref{Sync}, we  introduce our asynchronous and synchronous circuits performing accurate arithmetic operations, respectively. In Section \ref{Exp},
we give experimental results to evaluate the proposed circuits.  Section \ref{Conc} concludes this work with future directions. 
\section{Preliminaries} \label{limits}
We start with a few definitions. We define bit duration as the time duration of a single bit in a stream, and stream length as the number of bits in the stream. We define accuracy as an indicator of having correct or expected output values. For example, having an output value of 1/6 for a multiplier with input values of 1/2 and 1/3 is an accurate operation. For a circuit block, an adder or a multiplier in this study, it is fully-accurate if its operations are always correct, and it is semi-accurate if its operations are sometimes correct. 

Improving accuracy in SC has a fundamental limit as explained in the following theorem.

\begin{Theo}
Consider a system with ideal elements performing ideal SC. Accuracy of the system only depends on the expected output values $ z_{e} $'s and the output stream lengths $ n $.
\end{Theo} 

\begin{proof}
In SC, $ z_{e} $ can also be defined as the probability that each output bit takes a logic 1 value. Therefore, each output bit has a Bernoulli distribution and the output stream has Binomial distribution ($ p=z_{e} $). The standard error (standard deviation) and the relative error can be calculated as $ \sqrt{\frac{p\times(1-p)}{n}} $ and $\sqrt{\frac{(1-p)}{p\times n}} $, respectively; both only depend on $ p=z_{e} $ and $ n $.
\end{proof}


This relatively simple theorem tells us that 1) fully-accurate computation is impossible with SC that needs infinite stream lengths; 2) increasing stream lengths $X$ times results in a decrease in error values by only $\sqrt{X}$ times which is not efficient; and 3) in order to achieve high accuracy, randomness in output bit streams should be sacrificed.

Motivated by these inferences, we introduce a novel computing paradigm ``Bit Stream Computing (BSC)" constructed on the following three properties:
\begin{enumerate}
\item Unary bit streams having time series of 0's and 1's is used;
\item The value of a bit stream is calculated by the total number of 1 valued bits over the total number of bits in the stream; and
\item For a circuit or system employing BSC, same input values always result in same output values. In other words, output values are independent of the orderings or distributions of 1's and 0's in input bit streams unless their values do not change.
\end{enumerate}
Note that while SC satisfies the first two properties, bit serial computing with binary weighted bits \cite{hartley2012digit, denyer1985vlsi} satisfies the third property which is related to the accuracy of computing. As a result, BSC benefits from the logic used in SC, but does not necessarily employ randomly or Binomially distributed input/output bit streams as SC does. This allows fully-accurate computing with BSC.  

We perform accurate arithmetic addition and multiplication operations with BSC by considering the constraints given below. Note that since values of bit streams are in the range of 0-1, addition should be scaled by averaging the values. For example, addition of two input values $X_1$ and $X_2$ results in $\frac{X_1+X_2}{2}$. Throughout the paper we simply use the word ``addition" to refer ``scaled addition" . 

\begin{Lem}
	\label{lem_Add}
	Consider two input bit streams with lengths of $ n $. Suppose that the streams take $n+1$ values between $ 0/n$ and $ n/n $. Accurate addition of the inputs with BSC requires an output bit stream with a minimum length of $ 2\times n $.
\end{Lem}

\begin{proof}
Consider a worst case scenario for which one input takes the value of $ 1/n $ and the other one takes $ 0/n $. The output value should be $ 1/(2\times n) $ that requires a length of $ 2\times n $.
\end{proof}

\begin{Lem}
	\label{lem_Mult}
	Consider two input bit streams with lengths of $ n $. Suppose that the streams take $n+1$  values between $ 0/n$ and $ n/n $. Accurate multiplication of the inputs with BSC requires an output bit stream with a minimum length of $ n^{2} $.
\end{Lem}

\begin{proof}
	Consider a worst case scenario for which both inputs take the value of $ 1/n $.  The output value should be $ 1/n^{2} $ that requires a length of $ n^{2} $.  
\end{proof}

\begin{Theo}
	\label{thr_seccussive}
Consider a system performing BSC such that input and output stream lengths are $n$ and $m$, respectively where $n<m$. If the system's  current reaction time or delay is independent of the past, it cannot correctly process successive input streams.  
\end{Theo}

\begin{proof}
Suppose that the system has a delay of $d$ bits ($ d $ may be fractional).    It means that after applying input bit streams, the system needs to wait for  a time duration of $ d $ bits to have the first output bit. Consider two sets of successive input bit streams. After the completion of the first set, the system needs more time equivalent to $m-n+d$ bits to have the output in full. However, we know that after the time duration of $d$ bits, the output starts to have the results for the second input set. This is illustrated in Fig. \ref{fig_TimeSeriesDifferent}. As a result, $ m-n+d=d $, and $m=n$ should be satisfied to obtain correct results.
\end{proof}

\begin{figure}[!t]
	\centering
	\includegraphics{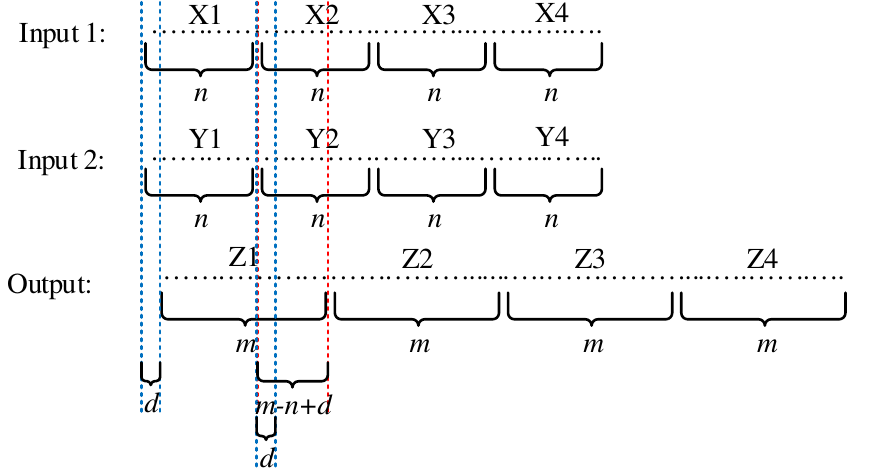}
	\vspace{-20pt}
	\caption{Demonstration of processing successive input bit streams.}
	\vspace{-0pt}
	\label{fig_TimeSeriesDifferent}
\end{figure}

Theorem \ref{thr_seccussive} leads to two solutions for successive processing of bit streams. The first one is controlling the system's delay sequentially. For example in Fig. \ref{fig_TimeSeriesDifferent}, for the first, the second, and the third set of input streams, the delay should be $d$, $d+m-n$, and $d+(2m-n)$ bits, respectively. Implementing such a complex and sequential system certainly kills the area advantage of BSC. The second solution is having same stream lengths for the inputs and outputs. This solution is much better not just for the area, but also for its suitability for multi-level designs. In this study, we use the second solution.

Using same stream lengths might result in inaccurate outputs. From Lemma \ref{lem_Add} and Lemma \ref{lem_Mult}, we know that we cannot achieve accurate addition and multiplication by using the same lengths if input bit streams are in full resolution meaning that they can take all possible values. For example, suppose that input and output stream lengths are 16, and multiplication is performed. If both inputs have values of 3/16, the correct result should be 9/256 or 0.5625/16, but we can only get either 0/16 or 1/16 from the output, so there is an error. To minimize the error, we round the output value to the nearest integer. In this example, the rounded result is 1/16.
Considering this accuracy issue, we classify the proposed asynchronous and synchronous circuits as semi-accurate with constant stream lengths and fully-accurate with increasing stream lengths. This is illustrated in Fig. \ref{fig_ProposedDesigns}. Note that fully-accurate ones are not proper for successive processing.
 
\begin{figure}[!t]
	\centering
	\includegraphics[scale=0.95]{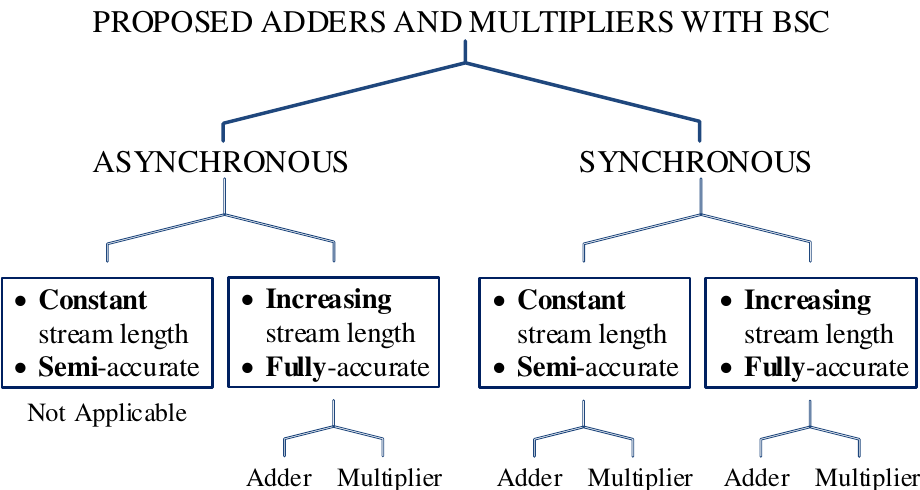}

	\caption{Summary of the proposed adder and multiplier designs.}

	\label{fig_ProposedDesigns}
\end{figure}
\section{Asynchronous Adders and Multipliers} \label{Async}

First, we clarify why we do not use constant stream lengths for asynchronous circuits as stated in Fig. \ref{fig_ProposedDesigns}. Different from fully-accurate adders and multipliers with increasing stream lengths, semi-accurate adders and multipliers with constant stream lengths do not always have a change in their output values if one of the input values changes. Therefore using a constant stream length needs  decision making of whether or not changing the output value or correspondingly whether or not changing the output bits. This requires to store and reuse of previously processed input bits in current operations, done with sequential circuits, and it is costly for asynchronous circuits. Therefore, we do not design asynchronous circuits with constant stream lengths. 

As an example, consider an adder performing BSC with a constant stream length of four. Suppose that applying input values 0/4 \& 0/4, 0/4 \& 1/4, and 1/4 \& 1/4 result in output values of 0/4, 0/4, and 1/4, respectively. Consider input streams 1,0,0,0 and 0,0,0,1. Here, the adder should store the information of the first received 1 without assigning any 1 to output bits until processing 
the last input bits, so there is a store and reuse operation. At the end the output stream becomes 0,0,0,1. 

For our designs we should also consider the following limitation.

\begin{Lem}
	\label{lem_Complexity}
	Consider a fully-accurate asynchronous system performing BSC such that it has input streams with  lengths of $n$ and an output stream with a length of $m$ where $n<m$. The system should consist of circuit elements having total of more than $ m-n $ outputs.  	
\end{Lem} 

\begin{proof}
At the time the output has its $ n $th bit, the remaining $ m-n $ bits should be kept in the system that necessitates $ m-n $ outputs excluding the output of the system. As an example, suppose that $n=2$ and $m=4$, and the system consists of inverters. Here, we need at least three outputs necessitating three inverters.
\end{proof}

Considering arithmetic operations and their stream length specifications previously given in Lemma \ref{lem_Add} and Lemma \ref{lem_Mult}, using Lemma \ref{lem_Complexity} we conclude that more than $ n $ and $ n^{2}-n $ outputs are needed for fully-accurate adders and multipliers, respectively, where $n$ is the length of input streams. Since we use inverters as circuit elements to achieve desired delay values,
more than $ n $ and $ n^{2}-n $ inverters are needed for  for fully-accurate adders and multipliers, respectively. Another restriction is that the number inverters should be even to prevent negation at the outputs. 

To select a proper inverter structure, we consider three criteria: 1) its circuit area should be small in harmony with the area advantage of BSC; 2) its rise and fall times should be small to preserve signal integrity; and 3) it should be controllable to compensate for changes in delay values. Considering different options, three inverter structures come forward, shown in Fig. \ref{fig_PnControlled}. Delay control of the conventional inverter, shown in Fig. \ref{fig_PnControlled} a), can be achieved by VDD scaling. For better control, VN and VP analog voltage inputs can be used as shown in Fig. \ref{fig_PnControlled} b). Additionally, to improve signal integrity, the inverter in \ref{fig_PnControlled} b) can be cascaded with a Schmitt trigger as shown in Fig. \ref{fig_PnControlled} c) \cite{mahapatra2002comparison}. Among these three options, we prefer the first conventional one for the sake of simplicity and area efficiency. 

\begin{figure}[!t]
	\centering
	\includegraphics{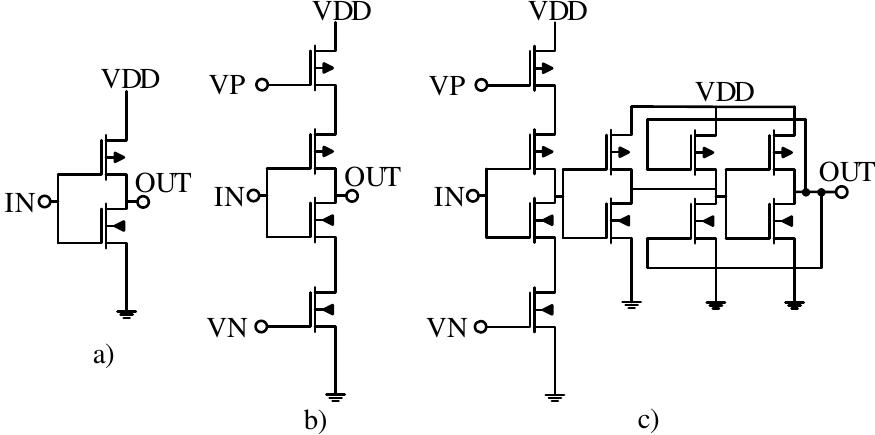}
		\vspace{-15pt}
	\caption{Inverter as a delay element: a) conventional inverter, b) NP-voltage controlled inverter, and c) its cascaded version with a Schmitt trigger.}
	\label{fig_PnControlled}
\end{figure}

\subsection{Increasing Stream Length: Fully-accurate Addition}

The proposed adder includes a delay block and an OR gate as shown in Fig. \ref{fig_AsyncAdd}. The delay block is used to postpone one of the inputs with a delay amount of the time duration of the input stream that can be calculated as (input stream length $n$)$\times$(bit duration). To satisfy this amount and  Lemma \ref{lem_Complexity}, we need to use at least $n+1$ inverters if $n$ is an odd number, and $n+2$ inverters if $n$ is an even number. As a result the area complexity is $O(n)$.

\begin{figure}[!t]
	\centering
	\includegraphics{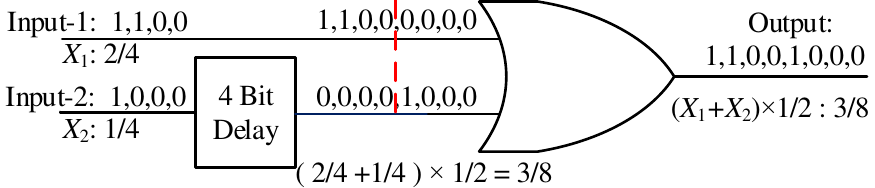}
	\caption{The proposed asynchronous adder with inputs having four bits ($n=4$).}
	\label{fig_AsyncAdd}
\end{figure}

\subsection{Increasing Stream Length: Fully-accurate Multiplication}

Accurate multiplication cannot be achieved unless each bit in one of the input streams is multiplied with each and every bit in the other one. Therefore, total of $ n^2 $ operations are needed where $n$ is the input stream length. We satisfy this by applying delays to the input streams for $2n-1$ different cases; $1$ case for no delay, $n-1$ cases for delaying one of the inputs more than the other one, and $n-1$ cases for the opposite. After making multiplications with AND gates, an OR gate is used to combine the results.  This is illustrated in Fig. \ref{fig_AsyncMultPre} for $n=3$. Here, there are total of 5 cases corresponding to 5 AND gates. Fig. \ref{fig_AsyncMult}  shows the circuit structure.

In general for $ n $ bit inputs, the circuit is constructed in four steps:
\begin{enumerate}

	\item The inputs are ANDed without any delay (corresponding to the AND gate numbered 1 in Fig. \ref{fig_AsyncMultPre}).
	\item Last $ n-i $ bits of Input-1 and first $n-i$  bits of  Input-2  are ANDed successively for $i=1, 2 ,...,n-1$, corresponding to the AND gates numbered 2 and 3 in Fig. \ref{fig_AsyncMultPre}. Total of $ n-1 $ AND gates are used for these operations with $n-1$ delay blocks for Input-1 and $n-1$ delay blocks for Input-2, corresponding to the 2 delay blocks in the upper part and the 2 delay blocks in the lower part of the circuit in Fig. \ref{fig_AsyncMult}.
	\item First $ n-i $ bits of Input-1 and last $n-i$  bits of  Input-2  are ANDed successively for $i=n-1, n-2,...,1$, corresponding to the AND gates numbered 4 and 5 in Fig. \ref{fig_AsyncMultPre}. Total of $ n-1 $ AND gates are used for these operations with $n-1$ delay blocks for Input-1 and $n-1$ delay blocks for Input-2, corresponding to the 2 delay blocks in the upper part and the 2 delay blocks in the lower part of the circuit in Fig. \ref{fig_AsyncMult}.
	
	\item Outputs of the all $ 2n-1 $ AND gates are ORed with a $ 2n-1 $ fan-in OR gate. The output of this OR gate gives the accurate result.
\end{enumerate}

Delay difference between the inputs of $ i $th and $ i-1 $th AND gates, representing the delay of the corresponding block, can be generalized as follows:
\[\text{For Input-1 }  \begin{cases}
	0 & i=1 \\
	n-(i-1)  & i=2,3,\dots,n \\
	n & i=n+1 \\
	i-(n+1) & i=n+2,n+3,\dots,2n-1
	\end{cases}
	\]

\[\text{For Input-2 }  \begin{cases}
	0 & i=1 \\
	n-(i-2)  & i=2,3,\dots,n \\
	-(n-2) & i=n+1 \\
	i-n & i=n+2,n+3,\dots,2n-1
	\end{cases}
	\]
Note that $ (n+1) $th case for the second input has a negative value  meaning  that it needs less delay than that of $ n $th case (see Fig. \ref{fig_AsyncMultPre} and Fig. \ref{fig_AsyncMult}).



\begin{figure}[!t]
	\centering
	\includegraphics{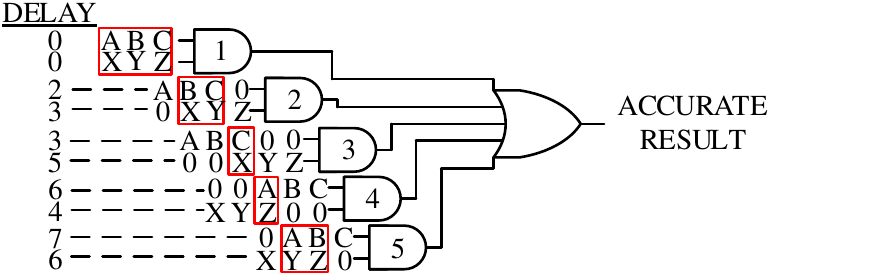}
	\caption{Elucidation of the proposed asynchronous multiplier for 3 bit inputs.}
	\label{fig_AsyncMultPre}
\end{figure}

\begin{figure}[!t]
	\centering
	\includegraphics{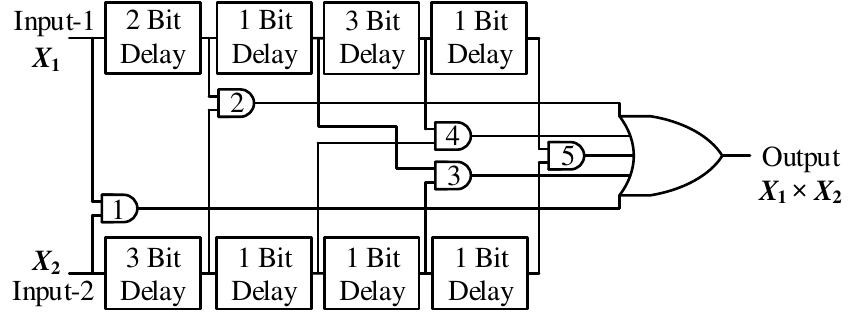}
	\caption{The circuit structure of the proposed asynchronous multiplier for 3 bit inputs.}
	\label{fig_AsyncMult}
\end{figure}

As a result, the delay blocks offer total of $n^2-(n-1)$ and $n^2-(n)$ bit delays for Input-1 and Input-2, respectively.  Therefore, we need at least $\approx 2n^2$ inverters to realize the delay blocks regarding that each delay block should have even number of inverters. The area complexity is $O(n^2)$.

\section{Synchronous Adders and Multipliers} \label{Sync}

The proposed asynchronous circuits are easy to design with delay controllability features. However, their area quickly grows with the input stream length $n$; for high $n$ values the circuits become inefficiently large.
 Also, they cannot process successive input streams. To solve these problems, we proceed to synchronous designs. The existence of auxiliary signals allows to keep and process the information carried by streams with binary digits.
 Indeed, the resulting circuits are hybrid with processing both streams and binary digits. 
 
We have two classes for the proposed synchronous designs that have increasing and constant stream lengths. While the former one is fully-accurate, similar to the proposed asynchronous ones, the latter one concedes slight errors with an important plus of being able to process successive input streams. As a result, we propose four adders and multipliers that are thoroughly explained in the following four subsections.

\subsection{Increasing Stream Length: Fully-accurate Addition}

We mainly use the same approach as we previously use for our asynchronous adder: one of the input streams waits until all bits of the other one is processed. Instead of using an asynchronous delay block as in Fig. \ref{fig_AsyncAdd}, we use synchronous blocks to store the input information in binary format that is more area efficient especially for large stream lengths. 

Fig. \ref{fig_SyncAccAdd} shows the proposed adder for an input stream length $n=8$; $ X_{1} $ and $ X_{2} $ represent the input values ranging between $ 0/8 $ and $ 8/8 $. Note that the stream length of the output should be $ 16 $ for accurate operation as previously stated in Lemma \ref{lem_Add}. The proposed adder first turns $X_1$ into a binary format via the counter. After the completion of the counting process, the register saves the information in the output of the counter. Then, binary to stream converting is done by the multiplexer. Finally, addition is performed with an OR gate. 

Even though the $ 8 $ bit stream corresponds to $ 3 $ bit binary resolution, the counter and the register are both selected $ 4 $ bit to be able to represent all 9 values between $ 0/8 $ and $ 8/8 $. Also, that is the reason why  OR gates are used before the multiplexer. The largest binary value coming from the outputs of the register $R_3 R_2 R_1 R_0$ is 1000 that should produce 1's for all bits in the stream at the output of the multiplexer. To do so, $R_3$ is ORed with other $R$'s.

For the multiplexer, along with the 4 inputs $ I_{0} $, ..., $ I_{3} $ coming from the register, there is one more input $ I_{4}$ used to make the output of the multiplexer logic 0 for a time duration of $n$ bits, needed to process the input stream coming from Input-2. Table \ref{tab_SyncAddMux} shows the relation between the data inputs, the selection inputs, and the output of the multiplexer. All selection inputs are actually clock signals with 50\% duty cycles. They can be generated from a single CLK input via frequency division with flip-flops, as shown in Fig. \ref{fig_FreqDiv}. Additionally, the $ TRIG $ input of the register is selected as $ \overline{S_{3}} $. 

If the input streams have $n$ bits, the counter and the register should be $ \log_2 n +1$ bit, and the multiplexer should have $ \log_2 n +2$ data inputs and $ \log_2 n +1$ selection inputs. Also, all auxiliary signals could be generated from a frequency divider circuit consisting of $ \log_2 n +1$ successive flip-flops. As a result, the area complexity is $O( \log n)$.
\begin{figure}[!t]
	\centering
	\includegraphics{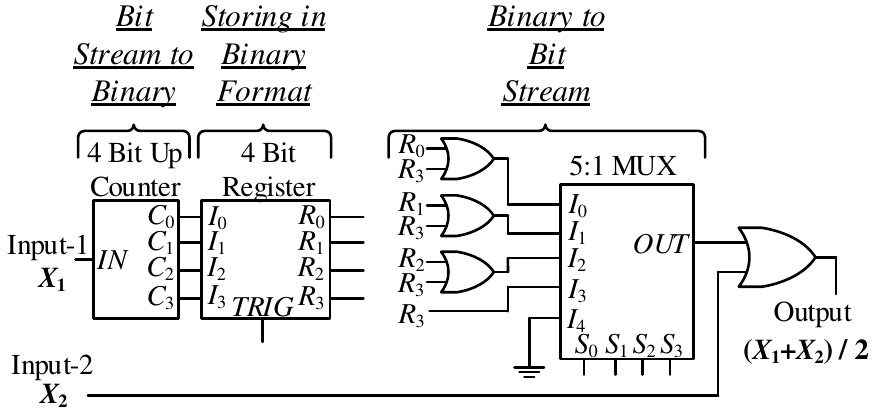}
	\vspace{-15pt}
	\caption{The proposed synchronous fully-accurate adder for 8 bit inputs.}
	\vspace{-5pt}
	\label{fig_SyncAccAdd}
\end{figure}

\begin{table}[!t]
	\caption{Relation between the selection inputs and the output of the 5:1 multiplexer in Fig. \ref{fig_SyncAccAdd}}
	\label{tab_SyncAddMux}
	\centering
	\begin{tabular}{|c|c|c|c||c|c|}
		\hline
		\multicolumn{4}{|c||}{\textbf{SELECTION INPUTS}}  & \multicolumn{2}{c|}{\textbf{OUTPUT}}  \\ 
		\hline
		\bm{$ S_3 $} & \bm{$ S_2 $} & \bm{$ S_1 $} & \bm{$ S_0 $} &  \textbf{Equivalence} & \textbf{Duration}\\ \hline
		0 & 0 & 0 & 0 &  $ I_{0} $ & 1 bit\\ \hline
		0 & 0 & 0 & 1 &  $ I_{3} $ & 1 bit\\ \hline
		0 & 0 & 1 & X &  $ I_{1} $ & 2 bit\\ \hline		
		0 & 1 & X & X &  $ I_{2} $ & 4 bit\\ \hline
		1 & X & X & X &  $ I_{4} $ & 8 bit\\ \hline				
	\end{tabular}
\end{table}

\begin{figure}[!t]
	\centering
	\includegraphics{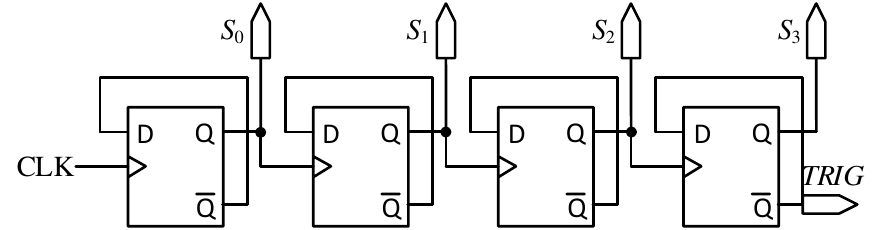}
			\vspace{-15pt}
	\caption{The generation circuitry of the auxiliary signals in Fig. \ref{fig_SyncAccAdd}}
		\vspace{-10pt}
	\label{fig_FreqDiv}
\end{figure}

\subsection{Increasing Stream Length: Fully-accurate Multiplication}

Consider two input bit streams with lengths of $n$. As mentioned earlier, accurate multiplication requires $ n^2 $ bitwise operations that is in compatible with  Lemma \ref{lem_Mult}. We satisfy this by repeating one of the streams $n$ times, and by repeating each bit of the other stream $n$ times. An example for $n=4$ is shown in Fig. \ref{fig_SyncAccMultPre}. Note that the orders of 0's and 1's in the input streams are not reflected to the repeated streams; after the counting process, we only have the information of the input values $ X_{1} $ and $ X_{2} $, not the orderings. In the example, both of the input streams are treated as (0,1,1,1) since $ X_{1} = X_{2} $.

 The circuit implementation of the proposed multiplier for $n=4$ is given in Fig. \ref{fig_SyncAccMult}. The counting and reconversion circuitry is nearly same with the one in Fig. \ref{fig_SyncAccAdd}. The only difference is the number of inputs in multiplexers (one less), because there is no need to wait for one of the streams as we do for the adder. The selection inputs of the upper multiplexer ($ S_{0} $, $ S_{1} $) are 4 times faster than those of the lower ones ($ S_{2} $, $ S_{3} $). Additionally, the \textit{TRIG-1} input can be selected as the negated form of $S_{2} $, and the \textit{TRIG-2} input is the negated form of the two times slowed version of $ S_{3} $. Therefore, all auxiliary inputs can be generated from a single clock signal by using a frequency divider circuit having 5 successive flip-flops.

Analyzing the scalability of the proposed multiplier, we see that the counters and registers should be $ \log_2 n +1$ bit, while the multiplexer should have $ \log_2 n +1$ inputs. Furthermore, $ \log_2 n +2$ successive flip-flops needed to generate all required auxiliary signals. Therefore, the area complexity is $O( \log n)$.

\begin{figure}[!t]
	\centering
	\includegraphics{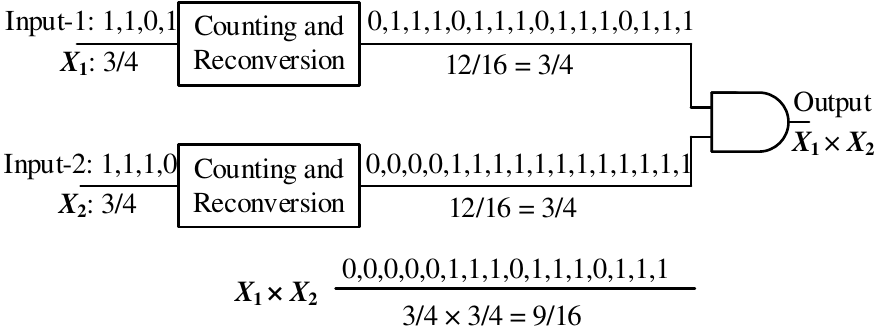}
	\vspace{-15pt}
	\caption{The proposed multiplication scheme for 4 bit inputs.} 
	\label{fig_SyncAccMultPre}
\end{figure}

\begin{figure}[!t]
	\centering
	\includegraphics{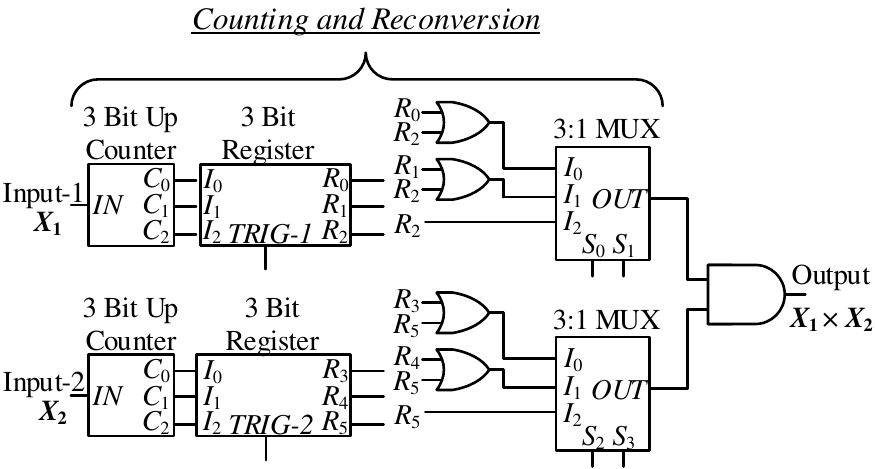}
	\vspace{-15pt}
	\caption{The proposed fully-accurate synchronous multiplier for 4 bit inputs.} 
	\label{fig_SyncAccMult}
\end{figure}

\subsection{Constant Stream Length: Semi-accurate Addition}

Addition in BSC is having an average of the input values $ X_{1}$ and $X_{2} $, and for the constant stream length, this can be performed with bit-by-bit averaging of the input stream bits. However, since the average of 1 valued and 0 valued bits results in 0.5 and it can not be represented with a single output bit, carry is needed to store the information. Table \ref{tab_OptSyncAdd} shows the truth table for such a solution. A circuit suiting Table \ref{tab_OptSyncAdd} should operate as desired. Fig. \ref{fig_InaccAdd} shows the circuit implementation of the proposed adder. It works as follows: if the input values $ X_{1}$ and $X_{2} $ are both even or both odd, then the result is correct; otherwise the result is the rounded version of the correct result with an error distance of $0.5/n$ where $n$ is the input stream length. Fig. \ref{fig_InaccAdd_ex} shows two examples for the proposed addition operation giving erroneous and accurate results. Since the circuit area  is constant, independent of $n$, the area complexity is $O(1)$.

Note that our circuit in Fig. \ref{fig_InaccAdd} has a quite similar performance compared to the scaled adder in \cite{lee2017energy}. However, with our point of view, we can generalize our adder for any $i$ number of inputs. 
In our design, at first the input bits are counted in parallel, and the result is added to the carry which has an initial value of $ i/2 $ to eliminate probable negative carry values. If the carry is larger than $ i $, the output becomes 1, and $ i $ is subtracted from the carry. Otherwise, the output is 0. Fig. \ref{fig_4bitSCSA} shows $i=4$ version where ``Parallel Counter" simply counts 1's in the input streams, and ``Binary Adder \& Output" first adds the current carry value to the output of the counter, then determines the output and updates the carry value with aforementioned process steps. Note that the structure coincides with modular parallel incrementers in the work \cite{parhami1995accumulative} through using its carry-out and sum as output and updated carry values, respectively.





\begin{table}[!tb]
\vspace{-5pt}
	\caption{Transition table of the proposed adder}
	\label{tab_OptSyncAdd}
	\centering
	\begin{tabular}{|c|c|c||c|c|}
		\hline
		\textbf{Carry} & \textbf{Input-1}  & \textbf{Input-2} & \textbf{Output}&  \textbf{Carry-new} \\ \hline
		X & 0 & 0 & 0 & Carry \\ \hline
		X & 0 & 1 & Carry & $ \overline{\mbox{Carry}} $\\ \hline	
		X & 1 & 0 & Carry & $ \overline{\mbox{Carry}} $\\ \hline
		X & 1 & 1 & 1 & Carry\\ \hline				
	\end{tabular}
\end{table} 

\begin{figure}[!t]
	\centering
	\includegraphics{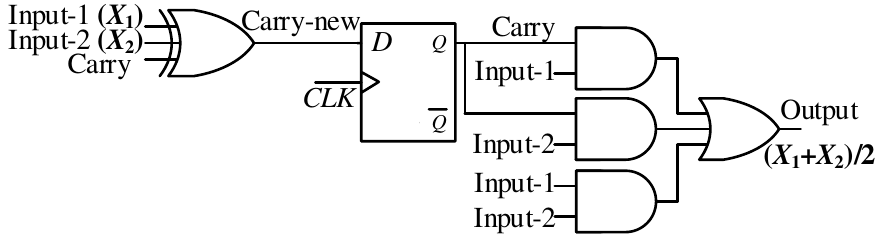}
	\caption{The proposed semi-accurate synchronous adder for two inputs.}
	\label{fig_InaccAdd}
\end{figure}

\begin{figure}[!t]
	\centering
	\includegraphics{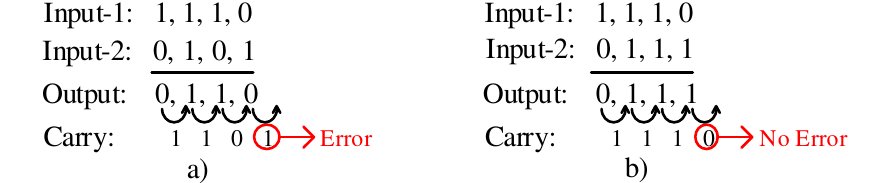}
	\caption{Examples for the proposed addition operation: a) $X_1=3/4$ and $X_2=2/4$, and b) $X_1=3/4$ and $X_2=3/4$.}
	\label{fig_InaccAdd_ex}
\end{figure}

\begin{figure}[!t]
	\centering
	\includegraphics{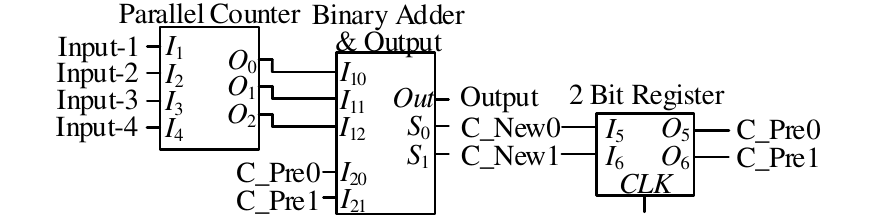}
	\caption{The proposed semi-accurate synchronous adder for four inputs.}
	\label{fig_4bitSCSA}
\end{figure}

\subsection{Constant Stream Length: Semi-accurate Multiplication}
%

We fundamentally use the same approach as we use for addition considering that multiplication is the repeated version of addition.  Suppose that input bit streams represent values $ X_{1}=a/n$ and $X_{2} =b/n$ where $n$ is the length of the streams. If we add $b$ copies of the first stream, or $a$ copies of the second, we can achieve multiplication by using bit-by-bit averaging with a carry. Different from the addition operation for which rounding to the nearest integer can be always satisfied with positive carry values, the multiplication operation with optimal error performance should have both positive and negative carry values between $-0.5n$ and $+0.5n$. Thus, we make the error distance upper bounded by $0.5/n$. Note that if we only used positive carries, this upper bound would be $1/n$.

Fig. \ref{fig_InaccMultEx} elucidates our multiplication scheme for 4 bit inputs. For example in Fig. \ref{fig_InaccMultEx} a), the first operation is adding three 1's with a result of 3; then $3/4$ is rounded to the nearest integer that is 1 as the output, and the carry with a value of -1 is transferred to the next bit operation for which three 1's and the carry results in $2/4$ that is rounded to 1 (it could have been rounded to 0 also) with a carry of -2 as an error.    

Instead of directly implementing the flow in Fig. \ref{fig_InaccMultEx} that requires to first process Input-2, and then depending on the value of it, process Input-1,  we process Input-1 and Input-2 separately. Thus, we can treat the inputs simultaneously by regenerating them with independent circuitries. We still satisfy the overall flow in Fig. \ref{fig_InaccMultEx} by achieving a faster, less complex, and smaller multiplier circuit.
In fact, regeneration of input signals to achieve better accuracy has been previously used in \cite{vahapoglu2016accurate} and \cite{jenson2016deterministic}. They manipulate the input streams to achieve 100\% accuracy at the output. However, these works produce larger output stream lengths than input ones that causes problems in processing successive input streams, as discussed in Section \ref{limits}.


For the proposed multiplication scheme, we regenerate input streams such that one of them is lined up in a way to  process  all 1's first and then all 0's, and the other one is used as a multiplicand. In every bit-by-bit operation, the multiplicand is added to the carry. Algorithm \ref{alg_OptMult} demonstrates the steps regenerating input streams of $ In1(n) $ and $ In2(n) $  as $ RegIn1(n) $ and $ RegIn2(n) $, respectively. The generation of $ RegIn1(n) $ is quite simple: first all 1's, then all 0's. However, $ RegIn2(n) $ is generated with a more complex way. For each iteration, number of 1's in $ In2(n) $ is added to $ Carry $. If new $ Carry $ is larger than or equal to $ n/2 $, it is subtracted by $ n $ and RegIn2(i) becomes 1, otherwise $RegIn2(i)$ becomes 0. Note that for the first iteration $ Carry =0$. 


\begin{algorithm} [!t]
\footnotesize
	\caption{Regeneration of Inputs for Optimal Error Performance in Constant Stream Multiplication}
	\label{alg_OptMult}	
	\begin{algorithmic}[1]
		\Procedure{RegIn}{In1(n),In2(n)}=(RegIn1(n),RegIn2(n)) \Comment{Regeneration of $ n $-bit Input Streams}
		\State $ SumIn1\gets sum(In1) $
		\State $ SumIn2\gets sum(In2) $
		\State $ Carry\gets 0 $
		\For {$ i\gets 1,n $}
			\If {$ i\le SumIn1 $} \Comment{Determining $ i $th bit of the 1st Regenerated Input }
				\State $ RegIn1(i)\gets 1 $
			\Else
				\State $ RegIn1(i)\gets 0 $
			\EndIf
			
			\State $ Carry\gets Carry+SumIn2 $ \Comment{Addition of Multiplicand to the Carry} \label{alg_OptMult_CarrySum}
			\If {$ Carry\geq n/2 $} \Comment{Determining $ i $th bit of the 2nd Regenerated Input } 
				\State $ RegIn2(i)\gets 1 $
				\State $ Carry\gets Carry-n $ \label{alg_OptMult_CarrySubt}
			\Else
				\State $ RegIn2(i)\gets 0 $
			\EndIf
		\EndFor
		\EndProcedure
	\end{algorithmic}
	\vspace{-5pt}
\end{algorithm}

Fig. \ref{fig_InaccMultPre} shows two examples for the regenerations of input bit streams. The multiplication operation is completed by AND operation of the regenerated streams. In Fig. \ref{fig_InaccMultPre} a), since the expected output value of $ 1/8 $ can be represented with 8 bits, the accurate output value of $1/8$ is obtained. On the other hand, in Fig. \ref{fig_InaccMultPre} b), the expected output value of $21/64$ cannot be represented accurately with 8 bits. Therefore, the output gives the nearest possible value of $3/8$. 

 \begin{figure}[!t]
 	\centering
 	\includegraphics{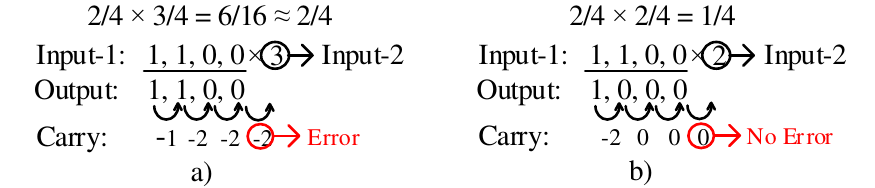}
 	\caption{Examples for the proposed multiplication approach: a) $X_1=2/4$ and $X_2=3/4$, and b) $X_1=2/4$ and $X_2=2/4$.}
 	\label{fig_InaccMultEx}
 \end{figure}

\begin{figure}[!t]
	\centering
	\includegraphics{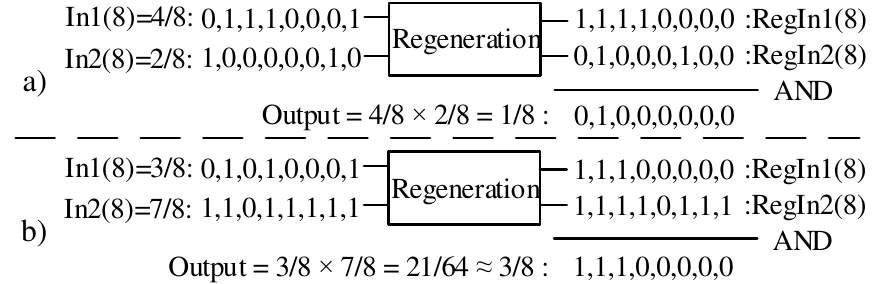}
	\vspace{-15pt}
	\caption{The regeneration of input streams with Algorithm \ref{alg_OptMult} with outputs giving a) no error, and b) optimal error.}

	\label{fig_InaccMultPre}
\end{figure}

Fig. \ref{fig_InaccMult} shows a circuit to realize Algorithm \ref{alg_OptMult}. It has three exactly same 4-bit up counters. Each  has an input $ IN_{i} $; four output ports $ C_{ij} $ and their negates $ C_{ijB} $; and four clear and four preset inputs $ CLR_{i} $ and $ PRE_{i}$, respectively. For simplicity, unused ports are  not generally shown in the circuit. That is why counters look different though they are exactly same. Similarly, each of the three identical registers has four inputs $ I_{ij} $; four outputs $ R_{ij} $ and their negates $ R_{ijB} $; and clear and clock inputs $ CLR $ and $ CLK $ ($ CLK $ or $ TRIG $), respectively. Again unused ports are not shown. 

Inputs $ IN_{1} $ and $ IN_{2} $ are first counted by up counters. The information of $ IN_{1} $ is saved in the 4-bit register and then inversely loaded to the next 4-bit up counter via the $ CLR_{i} $ and $ PRE_{i} $ inputs generated from $ TRIG $, $ R_{1j} $, and $ R_{1jB} $. The signals $ CLR_{i} $ and $ PRE_{i} $ are connected to the $ CLEAR $ and $ PRESET $ inputs of the corresponding D-FF in the counter to transfer the inversion of the saved information in the register to the counter after the counting in the first counter is completed. This means that the counter starts to count from inversion of saved information instead of from ``0000". Then the most significant bit (MSB) of the counter $ C_{23} $ and negated version of other bits, $ C_{2iB} $ for $ i=0,1,2 $, determine $ REG\_IN_{1} $. If the counter output is between 0111 and 1110, $ REG\_IN_{1} $ becomes 1; otherwise it becomes 0. For instance, for the case of $ X_{1}=5/8 $, the saved information is 0101 and the transferred information is 1010. So, the counter starts from 1010 and arrives at 1111 after 5 clock cycles, which means first 5 bits of $ REG\_IN_{1} $ are 1 and the rest are 0. Thus, we generate $ REG\_IN_{1} $ with an up counter and some logic circuit instead of a costly digital comparator. Note that in order to eliminate latency hit, regenerated signals are delayed by extra flip-flops. Eventually, the circuit gives its first output bit 2 clock cycles after getting the last input bits.

\begin{figure*}[!t]
	\centering
	\includegraphics{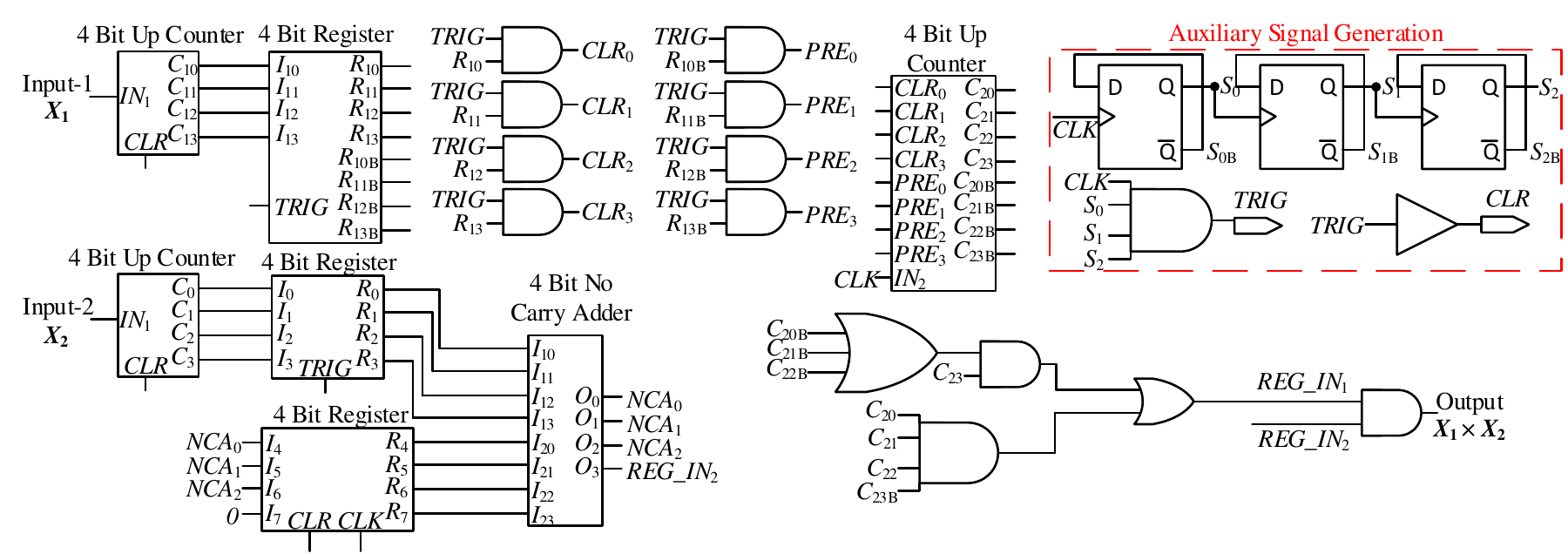}
	\vspace{-15pt}
	\caption{The proposed semi-accurate synchronous multiplier for 8 bit inputs.}
	\vspace{-5pt}
	\label{fig_InaccMult}
\end{figure*}

The saved information of $ IN_2 $ in the 4 bit register is summed with the previous sum in each clock cycle by a binary addition block so-called ``No Carry Adder" which excludes the MSB (or final carry) of a classical binary adder. This summation performs the operation $ Carry\gets Carry+SumIn2 $ in the line \ref{alg_OptMult_CarrySum} of Algorithm  \ref{alg_OptMult}. In other words, the inputs of the undermost register corresponds to $ Carry $. Note that $ Carry $ varies between $ [-n/2,n/2] $. To get rid of negative numbers, we start $ Carry $ from $ n/2 $, instead of 0, and shift the interval to $ [0,n] $ by starting the undermost register from 01..0, instead of from 00...0. Thus, if the MSB of the output of No Carry Adder is 1, which means the unshifted and shifted carries are larger than $ n/2 $ and $ n $, respectively, $ REG\_IN_{2} $ becomes 1. Otherwise, $ Carry $ is not large enough to produce 1, so $ REG\_IN_{2} $ becomes 0. Furthermore, the MSB of the input of the undermost register is always zero to implement the subtraction in the line \ref{alg_OptMult_CarrySubt} of Algorithm \ref{alg_OptMult}. If $ REG\_IN_{2} $ is 1 then it means that $ Carry\ge n $, so $ Carry $ should be subtracted by $ n $, i.e. MSB should be turned into 0. On the other hand, $ REG\_IN_{2} =0 $  does not require any change in $ Carry $, so the MSB should be again 0. 


The circuitry to generate auxiliary inputs $ TRIG $ and $ CLR $ is also added in Fig. \ref{fig_InaccMult}. Essentially, $ CLR $ is the slightly delayed version of $ TRIG $, because registers using $ TRIG $ signal need to save the outputs of the leftmost counters before they are cleared. The proposed design in Fig. \ref{fig_InaccMult} can be generalized for $ n $ bit inputs. Counters, registers, and binary adders should be $ \log_2 n +1 $ bit. Additionally, $ 2\times(\log_2 n +1 )$ AND gates are needed for producing $ CLR $ and $ PRE $ signals, and  $ log_2 n $ successive flip-flops are used in auxiliary signal generation. As a result, the area complexity becomes $O( \log n)$.

	\vspace{-5pt}
\section{Experimental Results} \label{Exp}
In this section, we evaluate the proposed six circuits:

\begin{itemize}  
\setlength{\itemindent}{-.15in}
\item \textbf{A}synchronous \textbf{I}ncreasing \textbf{S}tream-length \textbf{A}dder (\textbf{AISA}),

\item \textbf{A}synchronous \textbf{I}ncreasing \textbf{S}tream-length \textbf{M}ultiplier (\textbf{AISM}),

\item \textbf{S}ynchronous \textbf{I}ncreasing \textbf{S}tream-length \textbf{A}dder (\textbf{SISA}),

\item \textbf{S}ynchronous \textbf{I}ncreasing \textbf{S}tream-length \textbf{M}ultiplier (\textbf{SISM}),

\item \textbf{S}ynchronous \textbf{C}onstant \textbf{S}tream-length \textbf{A}dder (\textbf{SCSA}), and

\item \textbf{S}ynchronous \textbf{C}onstant \textbf{S}tream-length \textbf{M}ultiplier (\textbf{SCSM}).

\end{itemize} We present simulation results using the Cadence Genus tool with TSMC \SI{}{\textbf{0.18\micro m}} CMOS technology. The results are grouped in the following three subsections. In the first one, we thoroughly compare the proposed adders and multipliers with their predecessors in the literature. Comparisons are made in terms of area, speed, power, and accuracy. In the second subsection, we further evaluate the proposed circuits by considering timing problems and their effects on signal forms of output streams. We also test the circuits' abilities to be used in successive processing and multi-level designs. In the third subsection, we use the proposed adders and multipliers as well as their counterparts in the literature to implement a fully-connected neural network. Comparisons are made in terms of area and misclassification rates of the networks.

\subsection{Area, Speed, and Power Evaluations}

\begin{table*} [ht!]
	\caption{Performance comparison of adders}
		\vspace{-5pt}
	\label{tab_Adder_tc}
	\centering
	\begin{tabular}{|c|c|c||c|c|c|c||c|c|c|c|c||c|}
		\cline{3-13}
		                             \multicolumn{1}{c}{}                              &  \multicolumn{1}{c|}{}  & \multicolumn{1}{c||}{\textbf{Binary-to-Binary}} &                              \multicolumn{4}{c||}{\textbf{Binary-to-Stream}}                               &                                      \multicolumn{5}{c||}{\textbf{Stream-to-Stream}}                                      & \multicolumn{1}{c|}{\textbf{Analog-to-Stream}} \\ \cline{2-13}
		                            \multicolumn{1}{c|}{}                              & \textbf{Input Levels /} &             \textbf{Ripple Carry }              & \textbf{\cite{jenson2016deterministic}} & \textbf{\cite{gupta1988binary} } & \textbf{SISA} & \textbf{SCSA} & \textbf{\cite{jenson2016deterministic}} & \textbf{\cite{gupta1988binary}} & \textbf{AISA} & \textbf{SISA} & \textbf{SCSA} &        \textbf{\cite{najafi2017time} }         \\
		                            \multicolumn{1}{c|}{}                              & \textbf{Stream Length}  &                 \textbf{Adder}                  &                                         &                                  &               &               &                                         &                                 &               &               &               &                                                \\ \hline
		              \multirow{6}{*}{\rotatebox{90}{Area ($ um^{2} $)}}               &       \textbf{8}        &                       361                       &                  2006                   &               1128               &     1084      &     1247      &                  3798                   &              3441               &      83       &     1084      &      276      &                      153                       \\ \cline{2-13}
		                                                                               &       \textbf{16}       &                       594                       &                  2709                   &               1611               &     1417      &     1580      &                  5036                   &              4813               &      136      &     1417      &      276      &                      233                       \\ \cline{2-13}
		                                                                               &       \textbf{32}       &                       745                       &                  3401                   &               2119               &     1726      &     1991      &                  6331                   &              6529               &      242      &     1726      &      276      &                      393                       \\ \cline{2-13}
		                                                                               &       \textbf{64}       &                       978                       &                  4186                   &               2709               &     2059      &     2323      &                  7668                   &              8022               &      446      &     2059      &      276      &                      712                       \\ \cline{2-13}
		                                                                               &      \textbf{128}       &                      1222                       &                  4982                   &               3308               &     2403      &     2679      &                  9017                   &              9582               &      872      &     2403      &      276      &                      1351                      \\ \cline{2-13}
		                                                                               &      \textbf{256}       &                      1560                       &                  5744                   &               3951               &     2771      &     3413      &                  10377                  &              11245              &     1723      &     2771      &      276      &                      2628                      \\ \hline\hline
		\multirow{6}{*}{\rotatebox{90}{\parbox[c]{1.5cm}{\centering Max Freq. (GHz)}}} &       \textbf{8}        &                      3.10                       &                  0.93                   &               1.05               &     0.94      &     0.87      &                  0.90                   &              0.86               &     29.4      &     0.94      &     1.33      &                      23.5                       \\ \cline{2-13}
		                                                                               &       \textbf{16}       &                      2.26                       &                  0.88                   &               1.00               &     0.84      &     0.78      &                  0.88                   &              0.87               &     29.4      &     0.85      &     1.33      &                      26.1                       \\ \cline{2-13}
		                                                                               &       \textbf{32}       &                      2.17                       &                  0.83                   &               0.98               &     0.75      &     0.66      &                  0.83                   &              0.83               &     29.4      &     0.78      &     1.33      &                      27.7                      \\ \cline{2-13}
		                                                                               &       \textbf{64}       &                      1.73                       &                  0.81                   &               0.94               &     0.69      &     0.60      &                  0.79                   &              0.81               &     29.4      &     0.72      &     1.33      &                     28.5                      \\ \cline{2-13}
		                                                                               &      \textbf{128}       &                      1.71                       &                  0.81                   &               0.96               &     0.63      &     0.57      &                  0.78                   &              0.78               &     29.4      &     0.67      &     1.33      &                      29.0                      \\ \cline{2-13}
		                                                                               &      \textbf{256}       &                      1.69                       &                  0.79                   &               0.91               &     0.58      &     0.53      &                  0.77                   &              0.78               &     29.4      &     0.60      &     1.33     &                      29.2                      \\ \hline\hline		                                                                               
		             \multirow{6}{*}{\rotatebox{90}{Power ($ uW $)}}               &       \textbf{8}        &                       15.8                       &                  281                   &               265               &     185      &     271      &                  941                   &              993               &      2.66       &     511      &      60.9      &                      5.3                       \\ \cline{2-13}
		                                                                               &       \textbf{16}       &                       21.1                       &                  378                   &               340               &     247      &     315      &                  1153                   &              1315               &      4.58      &     593      &      60.9      &                      8.2                       \\ \cline{2-13}
		                                                                               &       \textbf{32}       &                       29.6                       &                  473                   &               423               &     283      &     363      &                  1400                   &              1612               &      8.41      &     712      &      60.9      &                      14                       \\ \cline{2-13}
		                                                                               &       \textbf{64}       &                       35.2                       &                  556                   &               513               &     336      &     400      &                  1638                   &              1984               &      16.1      &     815      &      60.9      &                      25.5                       \\ \cline{2-13}
		                                                                               &      \textbf{128}       &                      42.1                       &                  660                   &               609               &     367      &     426      &                  1874                   &              2256               &      31.4      &     965      &      60.9      &                      48.5                      \\ \cline{2-13}
		                                                                               &      \textbf{256}       &                      55.3                       &                  730                   &               686               &     403      &     518      &                  2127                  &              2573              &     62.1     &     1183      &      60.9      &                      94.5                      \\ \hline
	\end{tabular}
	\vspace{-10pt}
\end{table*}

\begin{table*}
	\caption{Performance comparison of multipliers}
		\vspace{-5pt}
	\label{tab_Multiplier_tc}
	\centering
	\begin{tabular}{|c|c|c||c|c|c|c||c|c|c|c|c||c|}
		\cline{3-13}
		                             \multicolumn{1}{c}{}                              &  \multicolumn{1}{c|}{}  & \multicolumn{1}{c||}{\textbf{Binary-to-Binary}} &                              \multicolumn{4}{c||}{\textbf{Binary-to-Stream}}                               &                                      \multicolumn{5}{c||}{\textbf{Stream-to-Stream}}                                       & \multicolumn{1}{c|}{\textbf{Analog-to-Stream}} \\ \cline{2-13}
		                            \multicolumn{1}{c|}{}                              & \textbf{Input Levels /} &                 \textbf{Array}                  & \textbf{\cite{jenson2016deterministic}} & \textbf{\cite{gupta1988binary}} & \textbf{SISM } & \textbf{SCSM} & \textbf{\cite{jenson2016deterministic} } & \textbf{\cite{gupta1988binary}} & \textbf{AISM} & \textbf{SISM} & \textbf{SCSM} &         \textbf{\cite{najafi2017time}}         \\
		                            \multicolumn{1}{c|}{}                              & \textbf{Stream Length}  &               \textbf{Multiplier}               &                                         &                                 &                &               &                                          &                                 &               &               &               &                                                \\ \hline
		              \multirow{6}{*}{\rotatebox{90}{Area ($ um^{2} $)}}               &       \textbf{8}        &                       943                       &                  1633                   &               953               &      1357      &     2328      &                   3426                   &              3266               &     2025      &     3253      &     4179      &                      313                       \\ \cline{2-13}
		                                                                               &       \textbf{16}       &                      2084                       &                  2197                   &              1402               &      1868      &     3098      &                   4524                   &              4603               &     7540      &     4317      &     5471      &                      978                       \\ \cline{2-13}
		                                                                               &       \textbf{32}       &                      3585                       &                  2784                   &              1863               &      2378      &     3790      &                   5714                   &              6273               &     28743     &     5505      &     6586      &                      2586                      \\ \cline{2-13}
		                                                                               &       \textbf{64}       &                      5424                       &                  3441                   &              2421               &      2889      &     4540      &                   6923                   &              7735               &    108953     &     6611      &     7938      &                     13911                      \\ \cline{2-13}
		                                                                               &      \textbf{128}       &                      7706                       &                  4074                   &              2972               &      3422      &     5345      &                   8109                   &              9246               &    435916     &     7740      &     9288      &                      ****                      \\ \cline{2-13}
		                                                                               &      \textbf{256}       &                      10407                      &                  4685                   &              3600               &      3980      &     6316      &                   9318                   &              10895              &    1740494    &     9158      &     10771     &                      ****                      \\ \hline\hline
		\multirow{6}{*}{\rotatebox{90}{\parbox[c]{1.5cm}{\centering Max Freq. (GHz)}}} &       \textbf{8}        &                      1.84                       &                  0.94                   &              1.12               &      0.94      &     0.85      &                   0.91                   &              0.86               &     14.7      &     0.81      &     0.84      &                      23.5                      \\ \cline{2-13}
		                                                                               &       \textbf{16}       &                      1.24                       &                  0.89                   &              1.08               &      0.92      &     0.75      &                   0.89                   &              0.87               &     14.7      &     0.77      &     0.81      &                      26.1                      \\ \cline{2-13}
		                                                                               &       \textbf{32}       &                      1.00                       &                  0.84                   &              1.06               &      0.84      &     0.77      &                   0.83                   &              0.83               &     14.7      &     0.73      &     0.76      &                      27.7                      \\ \cline{2-13}
		                                                                               &       \textbf{64}       &                      0.83                       &                  0.81                   &              1.04               &      0.80      &     0.76      &                   0.79                   &              0.81               &     14.7      &     0.71      &     0.75      &                      28.5                      \\ \cline{2-13}
		                                                                               &      \textbf{128}       &                      0.76                       &                  0.81                   &              1.00               &      0.74      &     0.76      &                   0.78                   &              0.78               &     14.7      &     0.68      &     0.74      &                      29.0                      \\ \cline{2-13}
		                                                                               &      \textbf{256}       &                      0.69                       &                  0.80                   &              0.99               &      0.68      &     0.73      &                   0.77                   &              0.78               &     14.7      &     0.62      &     0.73      &                      29.2                      \\ \hline\hline                                                                                                     		                                                                           
		               \multirow{6}{*}{\rotatebox{90}{Power ($ uW $)}}                 &       \textbf{8}        &                      19.4                       &                   302                   &               265               &      259      &     522      &                   916                   &              246               &     64.2      &     861      &     1192      &                      11.1                       \\ \cline{2-13}
		                                                                               &       \textbf{16}       &                      57.4                       &                   381                   &              340               &      341      &     660      &                   1178                   &              344               &     252      &     1085      &     1501      &                      35.1                       \\ \cline{2-13}
		                                                                               &       \textbf{32}       &                       125                       &                   482                   &              423               &      431      &     817      &                   1398                   &              406               &     996     &     1317      &     1765      &                      129                      \\ \cline{2-13}
		                                                                               &       \textbf{64}       &                       217                       &                   571                   &              513               &      501      &     895      &                   1631                   &              485               &    3918     &     1548      &     2127      &                     501                      \\ \cline{2-13}
		                                                                               &      \textbf{128}       &                       356                       &                   698                   &              609               &      600      &     1045      &                   1880                   &              587               &    15961     &     1769      &     2367      &                      ****                      \\ \cline{2-13}
		                                                                               &      \textbf{256}       &                       509                       &                   755                   &              686               &      661      &     1222      &                   2078                   &              643              &    62679    &     2100      &     2667     &                      ****                      \\ \hline
	\end{tabular}
	\vspace{-15pt}
\end{table*}

We evaluate the designs using Cadence Genus Synthesis Tool with TSMC 0.18 \SIUnitSymbolMicro m CMOS digital library. We make the tool automatically generate area, speed, and power results of the proposed and compared designs.
In comparisons, we consider three studies offering accurate stochastic operations that are based on clock division \cite{jenson2016deterministic}, using LFSR's \cite{gupta1988binary}, and PWM signals \cite{najafi2017time}. We also consider conventional binary ripple carry adder and array multiplier circuits. In order to make fair comparisons, we take into account the signal forms, classified as binary, stream, and analog, at the inputs and the outputs. All of the proposed six circuits with BSC use streams as inputs and outputs. However, since the proposed synchronous designs do already make stream-to-binary conversion via counters and registers, they can be directly used for binary-to-stream computing with even smaller circuit sizes. There is an exception for  Input-2 of SISA, whose overhead is also considered. Moreover, the studies \cite{jenson2016deterministic} and \cite{gupta1988binary} use binary inputs and stream outputs; to make them perform stream-to-stream computing, counters and registers can be added to the inputs, that is why their stream-to-stream cases are more costly than their binary-to-stream counterparts. 

On the other hand, since the study \cite{najafi2017time} uses analog inputs and it is not straightforward to make such conversions, we separately evaluate it. 
The designs in \cite{najafi2017time} essentially include analog comparators, ramp generators, and clock generators. With the assumption that clock generators predominate over the other blocks (especially for longer stream lengths), we prefer to synthesize them only. Clock generators are mainly inverter chains. Thus, for the designs in \cite{najafi2017time} and our asynchronous designs, we use shortest possible inverter chains regarding given stream lengths.

We report area, maximum frequency, and power results in Table \ref{tab_Adder_tc} and Table \ref{tab_Multiplier_tc} for  adders and  multipliers, respectively. Maximum frequency ($fmax$) numbers are directly related to worst case delay values to  represent how fast the design are. Note that one can obtain the minimum bit duration as $1/fmax$. Power values are generated using 500 MHz clock frequency at which all designs work properly. 
We consider different input levels; for example, the input level of 32 corresponds to 5 binary inputs or a stream having a length of 32. 

Examining the numbers in Table \ref{tab_Adder_tc}, we see that the proposed adders mostly overwhelm the others in their categories ``binary-to-stream" and ``stream-to-stream". Even for values in categories, ``binary-to-binary" and ``analog", there are always better values in our designs.  Similarly, in Table \ref{tab_Multiplier_tc}, our designs generally give the best results with an exception that in the category ``binary-to-stream", the study \cite{gupta1988binary} gives slightly better results.   

With regard to the power and speed performances of compared works, we can also comment on the energy consumptions. The studies [6], [10] and the proposed SISA and SISM circuits have longer output streams, i.e. longer processing times. Therefore, the consumed energy becomes crucially high, compared to SCSA and SCSM circuits having non-increasing output stream lengths. For instance, 256-stream-length SCSM's need 256 times less processing time and approximately 128 times less energy  
even in the worst case.

Note that $fmax$ values given in Table \ref{tab_Adder_tc} and Table \ref{tab_Multiplier_tc} are obtained using worst case delay conditions, so at this frequency values or smaller ones in one clock cycle one can guarantee correct logical functionality of the circuits. However, it does not mean that we always achieve correct output streams, obtained in many successive clock cycles, regarding the timing problems previously mentioned in Fig. \ref{fig_RealizationErrors}. For example, the problems given in Fig. \ref{fig_RealizationErrors} a) and c), happen in multi clock cycles, so they do not have any affect in obtaining $fmax$ values. Here, further timing investigations are needed and done in the following subsection.

\subsection{Timing Evaluations} 
\label{transistor_exp_results}

\begin{table*} [ht]
	\caption{Integrity and correctness results of the proposed adders and multipliers (INT:Integrity, COR:Correctness)}
		\vspace{0pt}
	\label{tab_Trans}
	\centering
	\begin{tabular}{|c|*{5}{c|c||}c|c|}
		\hline
		\multirow{ 2}{*}{\textbf{Bit Duration}} & \multicolumn{2}{c||}{\textbf{AISA}} & \multicolumn{2}{c||}{\textbf{AISM}} & \multicolumn{2}{c||}{\textbf{SISA}} & \multicolumn{2}{c||}{\textbf{SISM}} & \multicolumn{2}{c||}{\textbf{SCSA}} & \multicolumn{2}{c|}{\textbf{SCSM}}  \\ \cline{2-13}
		                                     & \textbf{ INT } &   \textbf{COR }    & \textbf{INT} &    \textbf{ COR}     & \textbf{ INT} &    \textbf{COR}     & \textbf{INT } &    \textbf{COR }    & \textbf{INT} &    \textbf{COR }     & \textbf{INT  } & \textbf{ COR     } \\ \hline
		           \textbf{0.5ns}            &     86\%     &       100\%        &    70.3\%    &        87.5\%         &    0\%     &       0\%        &      0\%      &         0\%         &    0\%    &         0\%         &      0\%       &        0\%         \\ \hline
		          \textbf{0.75ns}            &     85.5\%     &       100\%        &    64.7\%    &        81.1\%        &    0\%     &        62.5\%        &    71.6\%     &        75\%        &    0\%    &        0\%         &      54\%       &        50\%         \\ \hline
		            \textbf{1ns}             &     85.2\%     &       100\%        &    66\%    &         81.1\%         &    90.8\%     &        100\%        &    99.1\%     &        100\%        &    0\%    &        0\%         &      93\%       &        50\%         \\ \hline
		            \textbf{2ns}             &      N/A       &        N/A         &     N/A      &         N/A          &    97\%     &        100\%        &    99.6\%     &        100\%        &    98\%    &        100\%         &  80\%   &     100\%     \\ \hline
		           \textbf{10ns}             &      N/A       &        N/A         &     N/A      &         N/A          &    99.1\%     &        100\%        &    99.9\%     &        100\%        &    99.6\%    &        100\%         &  99.5\%   &       100\%        \\ \hline
	\end{tabular}
	\vspace{-15pt}
\end{table*}

\begin{table*} [ht]
	\caption{Integrity and correctness results of the proposed adders and multipliers with PVT simulations (INT:Integrity, COR:Correctness)}
	\vspace{0pt}
	\label{tab_Trans2}
	\centering
	\begin{tabular}{|c|*{5}{c|c||}c|c|}
		\hline
		\multirow{ 2}{*}{\textbf{Bit Duration}} & \multicolumn{2}{c||}{\textbf{AISA}} & \multicolumn{2}{c||}{\textbf{AISM}} & \multicolumn{2}{c||}{\textbf{SISA}} & \multicolumn{2}{c||}{\textbf{SISM}} & \multicolumn{2}{c||}{\textbf{SCSA}} & \multicolumn{2}{c|}{\textbf{SCSM}}  \\ \cline{2-13}
		                                     & \textbf{ INT } &   \textbf{COR }    & \textbf{INT} &    \textbf{ COR}     & \textbf{ INT} &    \textbf{COR}     & \textbf{INT } &    \textbf{COR }    & \textbf{INT} &    \textbf{COR }     & \textbf{INT  } & \textbf{ COR     } \\ \hline
		           \textbf{0.5ns}            &     77.7\%     &       87.5\%       &    50.7\%    &        68.8\%        &      0\%      &         0\%         &      0\%      &         0\%         &     0\%      &         0\%          &      0\%       &        0\%         \\ \hline
		          \textbf{0.75ns}            &     74.2\%     &       87.5\%       &    58.3\%    &         75\%         &      0\%      &         0\%         &      0\%      &         0\%         &     0\%      &         0\%          &      0\%      &        0\%        \\ \hline
		            \textbf{1ns}             &     72.4\%     &       87.5\%       &    23.7\%    &        37.5\%        &      0\%      &       62.5\%        &    6.25\%     &       62.5\%        &     0\%      &         0\%          &      0\%      &        0\%        \\ \hline
		            \textbf{2ns}             &      N/A       &        N/A         &     N/A      &         N/A          &    95.6\%     &        100\%        &    98.7\%     &        100\%        &    96.9\%    &        100\%         &      0\%      &       0\%        \\ \hline
		           \textbf{10ns}             &      N/A       &        N/A         &     N/A      &         N/A          &    99.1\%     &        100\%        &    99.9\%     &        100\%        &    99.4\%    &        100\%         &     99.1\%     &       100\%        \\ \hline
	\end{tabular}
	\vspace{-15pt}
\end{table*}

\begin{figure}[t!]
	\centering
	\includegraphics[scale=0.25]{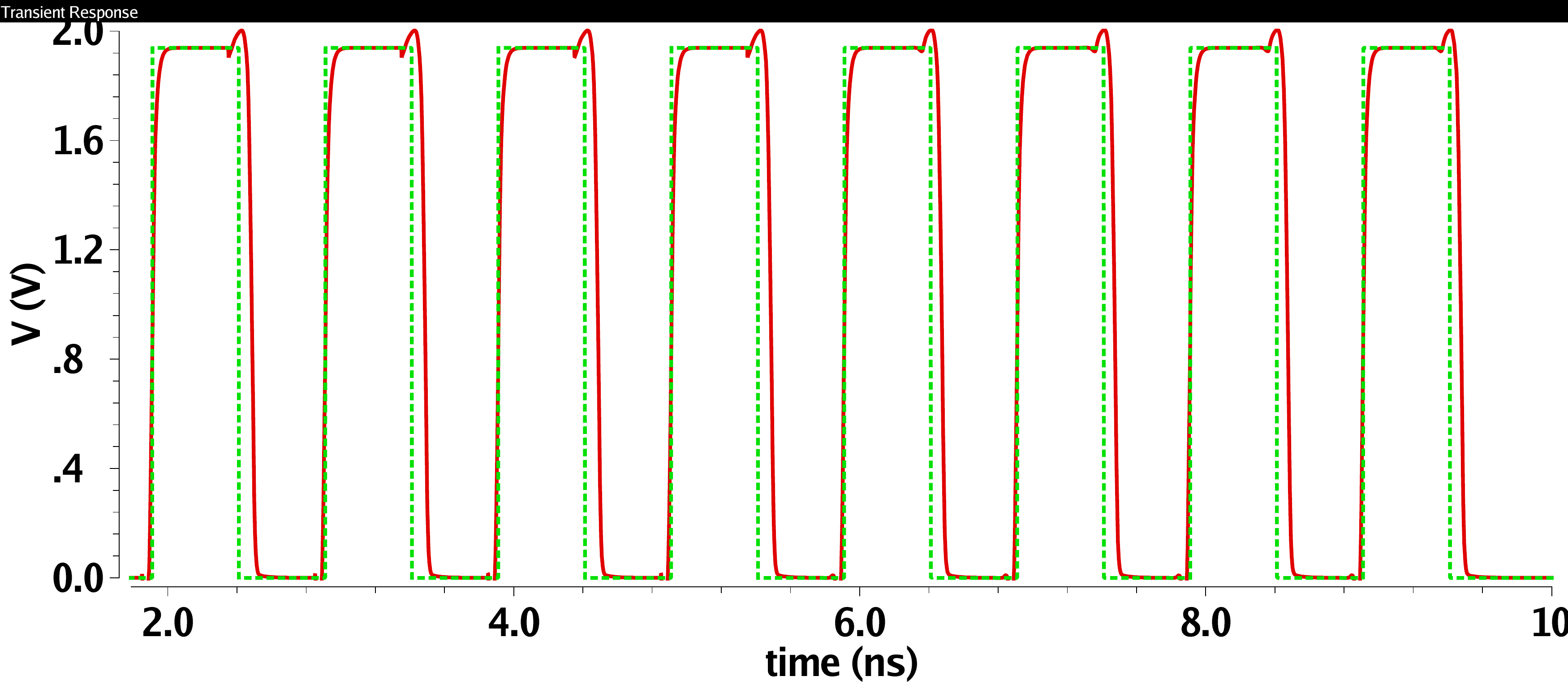}
{\small a)}
		\includegraphics[scale=0.25]{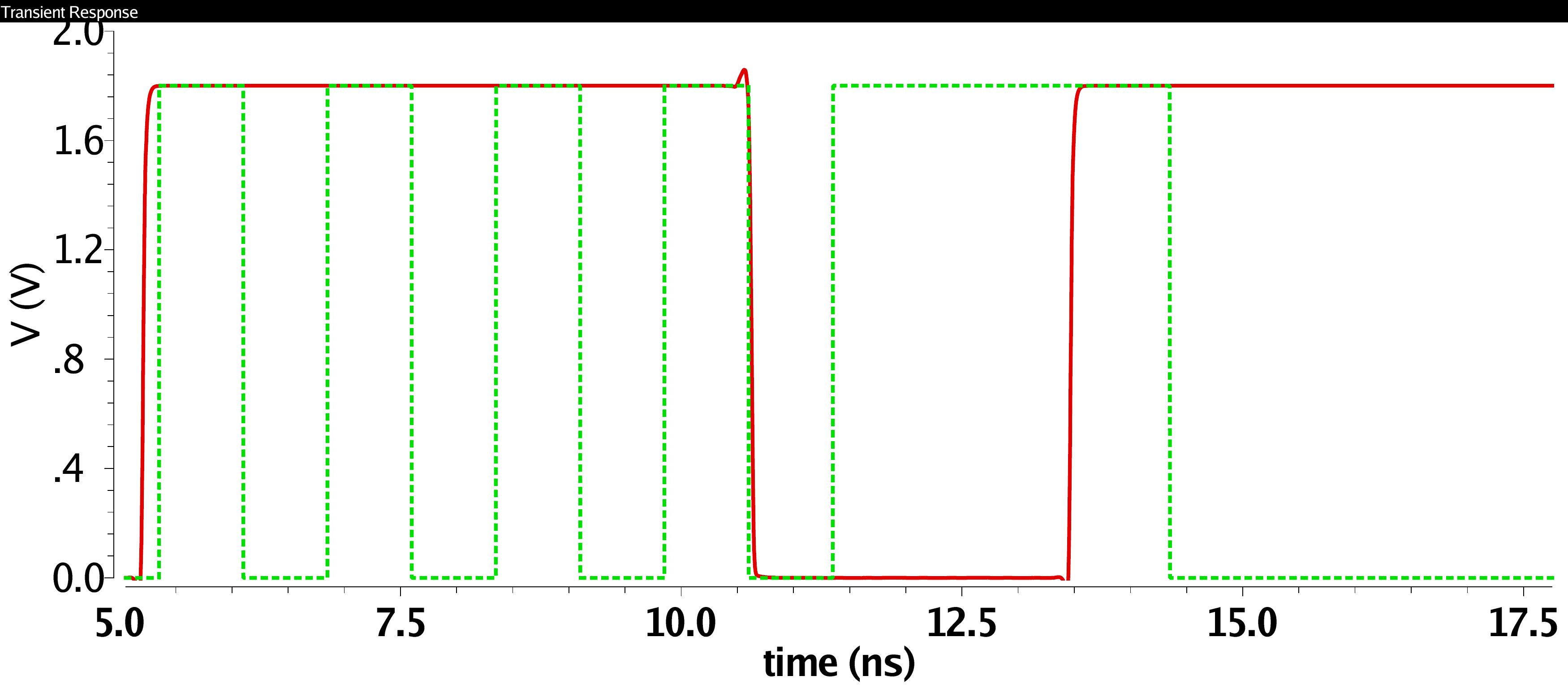}
		{\small b)}
		\vspace{-10pt}
	\caption{Output signal forms for a) AISA with 0.5ns bit duration, and b) SISA with 0.75ns bit duration. Red solid lines and green dashed lines represent real and expected outputs, respectively. }	
	\vspace{-10pt}
	\label{fig_TransLevel}
\end{figure}

We perform simulations with {0.18}\SI{}{\micro m} CMOS technology in Cadence environment. We test the proposed adders and multipliers for different input bit durations ranging between 0.5ns and 10ns; recall that a bit duration is defined as the time duration of a bit in  streams. Input values are selected such that the expected output value is around $1/2$ that can be considered as the worst case scenario for accuracy. Input stream lengths are selected as 8 for all simulations.
 Two performance metrics are considered: 1) integrity of the output signal to represent how much  the obtained timing durations of 1 valued bits match with the ideal ones, and 2) correctness of the obtained output value to represent how much the obtained values match with the expected output values. To elaborate, consider an expected output stream $1,0,1,1,0,0,0,1$ with a bit duration is 1ns. Also suppose that from simulations, we obtain 1 valued bit durations of 1.1ns, 2.1ns, and 0.4ns; ideally it should be 1ns, 2ns, and 1ns, respectively. For the first metric, we first calculate absolute deviations of 0.1ns, 0.1ns, and 0.6ns, then relative deviations 0.1ns$/$1ns, 0.1ns$/$2ns, and 0.6$/$1ns, and finally the average deviation of 25\% is obtained that results in 75\% signal integrity. For the second metric, we first obtain the ratios of the obtained output bit durations over the given bit duration. For this example, the ratios are 1.1ns$/$1ns, 2.1ns$/$2ns, and 0.4$/$1ns. Then we round them to the nearest integers, as 1, 2, and 0 for the example. Finally we have the obtained output value of $3/8$ with 75\% correctness. 

Table \ref{tab_Trans} shows the transistor level results of the proposed circuits for various bit duration values. Note that 0\% integrity values corresponds to 100\% or more deviations in durations, and they do not necessarily results in 0\% correctness values. Also note that 100\% integrity corresponds  to an ideal case, so it is not possible in our simulations. All of the proposed circuits can work for a bit duration equal to or smaller than 2ns, so the operating frequency of 500MHz can be achievable. Note that the proposed asynchronous circuits do a better job in small bit durations, thanks to their simple and delay block based structures. However, since increase in the bit duration requires more delay, and after a certain point it cannot be achieved with the controlling mechanism (in our case, VDD scaling), additional hardware in terms of extra inverters is needed. That is why for bit durations of 2ns and 10ns, the proposed circuits fail. Indeed we can generally claim that the proposed asynchronous and synchronous designs are proper for high speeds ($\geq$ 1GHz) and low speeds ($\leq$ 1GHz), respectively.

Fig. \ref{fig_TransLevel} shows the expected and the real output signal forms for three different cases. While the output signal has a good alignment in Fig. \ref{fig_TransLevel} a), Fig. \ref{fig_TransLevel} b) shows undesirable timing problems. Checking the corresponding values in Table \ref{tab_Trans}, we have 86\% integrity and 100\% correctness for the signal in Fig. \ref{fig_TransLevel} a), and 0\% integrity and 62.5\% correctness for the signal in Fig. \ref{fig_TransLevel} b).

We also carry out PVT (Process-Voltage-Temperature) simulations over the proposed designs. As corners, we use 5 different MOSFET process models, 2 different temperatures ($0^\circ$C and $80^\circ$C), and supply voltage values 10\% below and above the ideal one. Among 20 corners, we select the one with worst correctness numbers. Table \ref{tab_Trans2} lists the results. We see that even in the extreme conditions, the proposed circuits perform satisfactorily. While these conditions cause a decrease in correctness of the asynchronous designs, they force the synchronous designs to work in lower speeds.

Note that we evaluate the circuits for only 8 bit inputs as a simple case to show the potential of the proposed circuits for practical use. Larger input stream lengths could be considered. However, since larger lengths make our circuits more complicated with systematic timing design strategies needed to be followed, we consider this as a future work.    However, it is not hard to predict that for larger input stream lengths our asynchronous designs would have again quite successful outcomes, mainly because the architectures are quite straightforward. This is also true for SCSA, due to its scalable architecture. However, other synchronous designs, especially SCSM, circuits are getting more complex for larger stream lengths, so they may become more prone to timing errors.  

We also evaluate the aforementioned methods qualitatively  in Table \ref{tab_Qual}. Here, the latency criteria is the total processing time of receiving the output streams that is highly related to stream lengths. The conventional stochastic is the worst in latency considering that it needs relatively large stream lengths to obtain decent accuracy, previously shown in Fig. 2. The proposed circuits as well as the compared studies using deterministic streams perform better in this regard. However, the conventional binary is the best due to its parallel processing feature. For the accuracy criteria, all of the fully-accurate circuits do come first, followed by the proposed semi-accurate circuits. For the area criteria, we evaluate the circuits using the area values previously reported in Table III and IV. We rank the conventional stochastic as the best area efficient one. However if we considered the costly random number generators, needed for SC, then the area cost of SC would even become the worst.  The successive processing criteria is the ability of the circuits to produce same number of bits in an output as in an input. Here, the proposed constant stream circuits come forward. We also evaluate the suitability of the circuits to be used in multi-level designs. If an output of a circuit can be directly used as an input of another circuit in the next level, then the circuit is suitable for multi-level designs. As can be directly deducted from the definition of BSC in the preliminaries section, this requirement is met for any circuit using BSC. On the other hand, all of the compared studies using deterministic streams have limitations in this regard. 

In summary, when all criteria are equally important, the conventional binary and the proposed SCSA/SCSM are fairly competent. However, when accuracy and area are the most important factors, the proposed SISA/SISM comes forward.

\setlength\tabcolsep{1.5pt}

\begin{table}
	\caption{Qualitative comparison of adders and multipliers}
		\vspace{-5pt}
	\label{tab_Qual}
	\centering	
	\begin{tabular}{|c|*{5}{c|}}
		\cline{2-6}
		         \multicolumn{1}{c|}{}          & \textbf{Latency} & \textbf{Accuracy} & \textbf{Area} & \textbf{Successive} & \textbf{Multi-level} \\
		         \multicolumn{1}{c|}{}          &                  &                   &               & \textbf{Processing} &   \textbf{Design}    \\ \hline
		         \textbf{Conventional}          &    Excellent     &     Excellent     &   Moderate    &        Good         &      Excellent       \\
		           \textbf{Binary }             &                  &                   &               &                     &                      \\ \hline
		\textbf{\cite{jenson2016deterministic}} &     Moderate     &     Excellent     &     Poor      &        Poor         &       Moderate       \\ \hline
		   \textbf{\cite{gupta1988binary} }     &     Moderate     &     Excellent     &   Moderate    &        Poor         &       Moderate       \\ \hline
		    \textbf{\cite{najafi2017time}}      &     Moderate     &     Excellent     &     Poor      &        Poor         &         Poor         \\ \hline
		          \textbf{AISA/AISM }           &     Moderate     &     Excellent     &   Moderate    &        Poor         &      Excellent       \\ \hline
		          \textbf{SISA/SISM }           &     Moderate     &     Excellent     &     Good      &        Poor         &      Excellent       \\ \hline
		          \textbf{SCSA/SCSM }           &     Moderate     &       Good        &     Good      &      Excellent      &      Excellent       \\ \hline
		         \textbf{Conventional}          &       Poor       &       Poor        &   Excellent   &      Excellent      &      Excellent       \\
		          \textbf{Stochastic}           &                  &                   &               &                     &                      \\ \hline
	\end{tabular}
	\vspace{-15pt}
\end{table}

\setlength\tabcolsep{6pt}

\subsection{Evaluations within Neural Networks}
\begin{table*} [t]
	\caption{Total areas used in a neural network with PENDIGIT database ( Area:$ mm^{2} $, MR:Misclassification Rate)}
		\vspace{-7pt}
	\label{tab_Mnist}
	\centering
	\begin{tabular}{|c|*{5}{c|c||}c|c|}
		\hline
		\multirow{ 2}{*}{\textbf{Input Levels}} & \multicolumn{2}{c||}{\textbf{\cite{lee2017energy}}} & \multicolumn{2}{c||}{\textbf{\cite{gupta1988binary}}} & \multicolumn{2}{c||}{\textbf{SISA-SISM}} & \multicolumn{2}{c||}{\textbf{SCSA-SCSM}} & \multicolumn{2}{c||}{\textbf{Conventional Binary}} &            \multicolumn{2}{c|}{\textbf{Conventional Stochastic}}          \\ \cline{2-13}
		                                        & \textbf{ Area } &            \textbf{MR }             & \textbf{Area} &              \textbf{ MR}               & \textbf{ Area} &        \textbf{MR}        & \textbf{Area } &       \textbf{MR }        & \textbf{Area} &             \textbf{MR }             & \textbf{Area  } &                      \textbf{ MR     }                      \\ \hline
		              \textbf{8}                &     9.92      &               68.34\%               &    22.1     &                 7.52\%                  &     17.0     &          7.52\%           &     12.1     &          51.92\%           &    8.78     &                7.52\%                &     9.67    &                           84.6\%                           \\ \hline
		              \textbf{16}               &     12.6      &               37.56\%               &    31.3     &                 3.75\%                  &     23.8     &          3.75\%           &     15.6     &          12.58\%           &    13.8     &               3.75\%                &    12.3     &                           70.9\%                            \\ \hline
		              \textbf{32}               &     15.4      &               13.27\%                &    40.3     &                 2.94\%                  &     29.9     &          2.94\%           &     19.2     &          5.20\%           &    19.8     &                2.94\%                &    15.1     &                           52.2\%                           \\ \hline
		              \textbf{64}               &     18.5      &               5.21\%                &    50.9     &                 3.06\%                  &     36.9     &          3.06\%           &     23.0     &          3.37\%           &    27.0     &                3.06\%                &    18.2     & 31.3\%                \\ \hline
		             \textbf{128} & 21.6 & 3.59\% & 61.7 & 2.74\% & 44.2 &2.74\% & 27.0 & 3.09\% & 35.4 & 2.74\% & 21.4 & 15.12\%                \\ \hline
		             \textbf{256} & 24.2 & 3.35\% & 73.5 & 2.77\% & 51.7 & 2.77\% & 31.2 & 2.94\% & 44.7 & 2.77\% & 23.9 & 7.72\%              \\ \hline
	\end{tabular}
\end{table*}

To further evaluate the proposed circuits, we choose a fully-connected neural network because it mainly consists of adders and multipliers, and it does not require perfect accuracy. We use the PENDIGIT database which is a set of handwritten digits \cite{alpaydin1998pen}. It has 16 different input features corresponding to 16 perceptrons in the input layer of the neural network. Also our network has one hidden layer having 100 perceptrons. Obviously the output layer has 10 perceptrons to represent 10 digits. Except for the conventional implementation of the network with binary multipliers and adders, in each layer binary input values and their weights are multiplied with the binary-to-stream multipliers, and then they are summed in pairs with stream-to-stream adders. After that by first using stream-to-binary converters (just counters), all processes up to the next layer are implemented with conventional binary circuits. We also use rectifier linear unit (ReLU) as an activation function, due to its quite simpler hardware implementation and similar accuracy performance compared to sigmoid and others. In training the network, we use exact or fully-accurate adders, multipliers, and converters as well as the RELU. Therefore, weights of the network is same for all different implementations involving different adder and multiplier structures.

In comparisons, we consider four different implementation techniques of adders and multipliers. The first one offers the most area efficient accurate adders and multipliers among the studies considered in Table \ref{tab_Adder_tc} and Table \ref{tab_Multiplier_tc} \cite{gupta1988binary}. The second one uses adders very similar to the proposed SCSA, and conventional stochastic multipliers (AND gates) \cite{lee2017energy}. The third one uses conventional binary ripple carry adders and array multipliers, and finally the fourth technique employs conventional stochastic adders (2-to-1 multiplexers) and multipliers (AND gates). 
For the conventional stochastic circuits, randomly distributed input streams are needed. Therefore, binary-to-stream multipliers need one LFSR and digital comparators as twice as the multipliers. As stated in \cite{gaines1967stochastic}, one LFSR is enough for input streams, thanks to the low correlations between shifted streams. Additionally, each stream-to-stream adder needs a digital comparator for the generation of 0.5 valued stream for the select input of the multiplexer; again one LFSR is adequate for all adders.

Table \ref{tab_Mnist} gives the results. The proposed implementations are clearly the best ones in terms of accuracy and circuit area. Of course, if we did not consider the costs of stochastic number generators, \cite{lee2017energy} and the conventional stochastic would have much smaller transistor counts. However, this would not be a fair comparison.

	\vspace{-10pt}
\section{Conclusion} \label{Conc}

We introduce a novel computing paradigm ``Bit Stream Computing (BSC)" that benefits from the area advantage of stochastic logic and the accuracy advantage of conventional binary logic. We implement accurate arithmetic multiplier and adder circuits with BSC as well as using them in a neural network. Experimental results performed in Cadence environment with \SI{}{0.18\micro m} CMOS technology approve the efficiency of the proposed circuits. As a future work, we aim to develop hybrid computing schemes performing BSC and conventional binary computing. We will also investigate bit/digit serial computing in this regard.

Another direction is testing the proposed circuits in large area electronics including organic and flexible circuits that should have relatively small number of transistors. We comment that as opposed to conventional binary circuits, the proposed circuits performing BSC can be suitable for this.

\bibliographystyle{IEEEtran}
	\vspace{-7pt}
\bibliography{main}
	\vspace{-20pt}
	
 \begin{IEEEbiography}[{\includegraphics[width=1in,height=1.25in,clip,keepaspectratio]{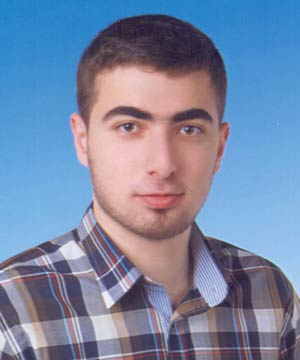}}]{Ensar Vahapoglu}
	received the B.Sc. degree from the Department of Electronics and Communication Engineering Istanbul Technical University, Istanbul, Turkey in 2015. He is currently a M.Sc. student and works as a research assistant in the same department. His main research areas are analog/digital circuits design, stochastic computing and quantum computing. 
	\end{IEEEbiography}
	\vspace{-20pt}
\begin{IEEEbiography}[{\includegraphics[width=1in,height=1.25in,clip,keepaspectratio]{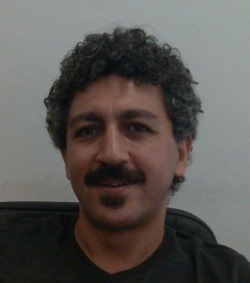}}]{Mustafa Altun}
received   his   BSc   and   MSc degrees  in  electronics  engineering  at  Istanbul Technical University in 2004 and 2007, respectively. He received his PhD degree in electrical engineering with a PhD minor in mathematics at the University of Minnesota in 2012. Since 2013, he has served as an assistant professor at Istanbul Technical University and runs the Emerging Circuits and Computation (ECC) Group. Dr. Altun has been served as a principal investigator/researcher  of  various  projects  including  EU H2020 RISE, National Science Foundation of USA (NSF) and TUBITAK projects. He is an author of more than 50 peer reviewed papers and a book chapter, and the recipient of the TUBITAK Success, TUBITAK Career, and Werner von Siemens Excellence awards.
\end{IEEEbiography}

\noindent \textbf{List of Differences:}

A preliminary version of this paper, titled “Accurate Synthesis of Arithmetic Operations with Stochastic Logic” was presented at the IEEE Computer Society Annual Symposium on VLSI (ISVLSI), 2016. Nearly 70\% of material in this manuscript is new:
\begin{itemize}
	\item At least 50\% of the material in the introduction section (Section I) including the entire Subsection I-A;
	
	\item At least 40\% of the material in the preliminaries section (Section II) including Theorem II and the summary of the proposed circuits, corresponding to Fig. 5 and the related text;
	
	\item Among the proposed 3 adders and 3 multipliers, given in Section  III and IV, only the asynchronous adder is fundamentally same as the circuit proposed in the previous publication; the other 5 circuits are completely new;
	\item The experimental results section (Section V) is almost 100\% new. All results for the adders and multipliers in Section V-A and V-B as well as those of the neural network application are new.
\end{itemize}

\end{document}